\documentclass[a4paper,english,11pt]{scrartcl}
\usepackage[left=2.5cm, right=2.5cm, top=2.5cm, bottom=2.5cm]{geometry} % ESA
\usepackage{tikz}
\usepackage[utf8]{inputenc}
\usepackage[T1]{fontenc}

\usepackage[title]{appendix}
\usepackage[affil-it]{authblk}

\usepackage{amsmath}
\usepackage{amssymb}
\usepackage{amsthm}
\usepackage{pst-node}
\usepackage{url,graphics,epsfig}
\usepackage{tabularx,subfigure}
\usepackage{booktabs,framed}
\usepackage[shortlabels]{enumitem}	% verbesserte Aufzählungen

\usepackage{xspace,url}  % should be last

\usepackage[bookmarks=false,colorlinks=true, linkcolor=blue, urlcolor=blue, citecolor=blue, breaklinks=true]{hyperref}
\usepackage{cleveref}

\newcommand{\R}{\mathbb{R}}

\newcommand{\JWvR}{$\text{JWVR}$}
\newcommand{\klammer}[1]{(   #1  )}
\newcommand{\BR}{\text{BR}}
\newcommand{\dt}{\text{d}t}
\newcommand{\eps}{\varepsilon}
\newcommand{\norm}[1]{||#1||}
\DeclareMathOperator{\conv}{conv}
\newcommand{\PNE}{\text{PNE}}
\newcommand{\OPT}{\text{OPT}}
\newcommand{\PoS}{\text{PoS}}
\newcommand{\PoA}{\text{PoA}}
\newcommand{\twelvevdots}{%
  \vbox{\baselineskip1ex\lineskiplimit0pt%
 \hbox{.}\hbox{.} \hbox{.}\hbox{.}\hbox{.}\hbox{.}\hbox{.}\hbox{.}\hbox{.}\hbox{.}\hbox{.}\hbox{.}\hbox{.}\hbox{.}\hbox{.}\hbox{.}\hbox{.}\hbox{.}}}

\theoremstyle{plain}
\newtheorem{theorem}{Theorem}[section]

\newtheorem{lemma}[theorem]{Lemma}

\theoremstyle{definition}

\newtheorem{definition}[theorem]{Definition}

\newtheorem{example}[theorem]{Example}

\title{Capacity and Price Competition \\ in Markets with Congestion Effects}
\author{Tobias Harks and Anja Schedel}
\affil{\small Universit\"at Augsburg, Institut für Mathematik, Germany\\
\href{mailto:tobias.harks@math.uni-augsburg.de}{\{\texttt{tobias.harks,anja.schedel\}@math.uni-augsburg.de}}}
\begin{document}

\maketitle

\begin{abstract}
We study oligopolistic competition in service markets where
firms offer a service to customers. The service quality of a firm 
-- from the perspective of a customer -- depends on the level
of congestion and the charged price. A firm can set a price
for the service offered and additionally decides on the service capacity in order to mitigate congestion. The total profit of a firm is  derived from the gained revenue minus the capacity investment cost. 
Firms \emph{simultaneously} set capacities and prices
in order to maximize their profit and customers \emph{subsequently} choose
the services with lowest combined cost (congestion and price). 
For this basic model, Johari, Weintraub and Van Roy~\cite{JohariWR10} 
derived the first existence and uniqueness results of pure Nash
equilibria (PNE) assuming mild conditions on congestion functions.
Their existence proof relies on Kakutani's fixed-point theorem and a key assumption for the theorem to work is that demand for service is \emph{elastic} 
(modeled by a smooth and strictly decreasing inverse demand function).

In this paper, we consider the case of \emph{perfectly inelastic demand}, that is, there is a \emph{fixed} volume of customers requesting service. 
This scenario applies to realistic cases where customers are not willing to drop out of the market, e.g. if prices are regulated by reasonable price caps.
We investigate existence, uniqueness and quality of PNE for models with inelastic demand and price caps. 
We show that for linear congestion cost functions, there exists a PNE. This result requires a completely new proof approach compared to previous approaches, 
since the best response correspondences of firms may be empty, thus standard fixed-point arguments are not directly applicable.  
We show that the game is $C$-secure (see Reny~\cite{Reny99} and McLennan, Monteiro and Tourky~\cite{McLennan11}), 
which leads to the existence of PNE. We furthermore show that the PNE is unique, and that the efficiency compared to a social optimum is unbounded in general. 
  \end{abstract}
\newpage

\section{Introduction}
We consider a model for oligopolistic capacity and price 
competition in service markets exhibiting congestion effects.
There is a set of firms offering a service to customers
and the overall customer utility depends on the congestion
level experienced and the  price charged. 
In addition to setting a price, a firm can invest in
installing capacity to reduce congestion.
Firms compete against
each other in the market by \emph{simultaneously} choosing
capacity and price levels while
customers react \emph{subsequently} choosing
the most attractive firms (in terms of 
congestion and price levels).
This model captures many aspects
of realistic oligopolistic markets
such as road pricing in traffic
networks (see~\cite{XiaoYH07} and~\cite{Zhang2004}) and competition
among WIFI-providers 
or cloud computing platforms (see~\cite{AcOz07} and~\cite{Anselmi:2017}).

In a landmark paper, Johari, Weintraub and Van Roy~\cite{JohariWR10} (\JWvR\ for short)
studied the fundamental question of existence, uniqueness and worst-case 
quality of pure Nash equilibria for this model. They derived several existence and uniqueness results under fairly general assumptions on the functional
form of the congestion functions. 
Specifically, they showed that
for models with \emph{elastic demand} (and some concavity assumptions)
a pure Nash equilibrium exists and is essentially unique.
The existence result is based on the Kakutani 
fixed-point theorem (see~\cite{Kakutani41} and the generalizations of~\cite{Fan52} and \cite{Glicksberg52}) and crucially
exploits the assumption that demand is elastic: 
The elasticity is modeled
by a smooth and strictly decreasing inverse demand function mapping each volume of
customers to the combined cost (latency and price) under which the given volume of
customers is participating in the game. 
This way, the best response correspondences of firms are always non-empty,   
allowing for the application of Kakutani's fixed-point theorem. 

In this article, we focus on the case of \emph{inelastic demand}, 
that is, there is a fixed volume of customers
requesting service. 
 Arguably much less is known
in terms of existence and uniqueness of
equilibria in this case although many works in the transportation science and algorithmic game theory
community (see, e.g., \cite{Bergendorff97congestion,DialpartI,Dial99I,Dial99II,Yang04,Fleischer04} and \cite{Roughgarden02,CCS06, Ackermann07,Fabrikant04}, respectively) assume inelastic demand and this case is usually considered as
 fundamental base case. As we show in this paper, in terms of equilibrium existence,
the case of  inelastic demand is much more complicated  
 compared to the seemingly more general case of elastic demand: 
best responses do not always exist 
 in the ``inelastic case'' 
putting standard fixed-point approaches out of reach. 
Besides this theoretical aspect, the case of inelastic demand is also interesting from a practical point of view. 
Litman~\cite{Litman12} discusses various factors influencing travel demand, and summarizes several transport
elasticity studies. Not surprising, for example the availability and quality of alternatives,
the reason for travel, or the type of traveller, impact the elasticity. If there are few alternatives,
as it is for example the case in Scandinavian freight transport, where rail and
ship-based freight networks are sparse compared to road networks, demand tends to be
inelastic~(\cite{Rich11}). Further examples where demand tends to be less elastic are higher value
travel (as business or commute travel, in particular ``urban peak-period trips''), or travel
of people having higher income, see~\cite{Litman12}.

\subsection{Related Work}
\JWvR\ derived several existence and uniqueness
results for pure Nash equilibria assuming elastic demand
of customers. The elasticity is modeled by a differentiable
strictly decreasing inverse demand function. Depending on the
generality of the allowed congestion cost functions,  further concavity assumptions
on the inverse demand (or demand) function are imposed by \JWvR.
The existence proof is based on Kakutani's fixed-point theorem which crucially
exploits that the best response correspondences of the firms are non-empty. 

For the model with inelastic demand -- that we consider in this paper -- 
\JWvR\ derive an existence result assuming
\emph{homogeneous firms} (that is, all firms have the same congestion cost function)  together with some further assumptions on the congestion costs and the number
of firms. As shown by \JWvR, homogeneity (together with some assumptions on
the congestion costs) implies
that there is only one symmetric
equilibrium candidate profile.  For this specific symmetric strategy profile, they
directly prove stability using concavity arguments. This proof technique is clearly
not applicable in the general non-homogeneous case (with inelastic demand).
As already mentioned, a further  difficulty  stems from the fact
that one cannot directly apply standard fixed-point arguments as even a best response
of a firm might not exist. 

Acemoglu et al.~\cite{AcemogluBO2009} study a capacity and price 
competition game assuming that the capacities
represent ``hard'' capacities bounding the admissible customer volume
for a firm. They observe that pure Nash equilibria
do not exist and subsequently study a two-stage
model, where in the first stage firms determine capacities
and in a second stage they set prices. For this model,
they investigate the existence and worst-case efficiency of equilibria (see
also~\cite{KrepsS83} for earlier work on the two-stage model).
Further models in which capacities or prices are determined \emph{centrally} in order to reduce the total travel time of the resulting Wardrop equilibria (plus investments for the case of capacities) have been considered for some time already,
 see for example~\cite{Marcotte85mp,GairingHK17} for setting capacities, and~\cite{Beckmann56,Yang04} for setting prices. 
In a very recent work, Schmand et al.~\cite{Schmand19} study existence, uniqueness and quality of equilibria in a network investment game in which the firms invest in edges of a series-parallel graph (but do not directly set prices).

\subsection{Our Results and Proof Techniques}
We study in this paper oligopolistic capacity and price
competition assuming that
there is a \emph{fixed} population of customers 
requesting service. This assumption applies
to road pricing games,
where customers usually own a car  
and are not willing to opt out of the game 
given that  prices range in realistically bounded domains.
Consequently, our model allows for the possibility of a~priori upper bounds with respect to
the prices set by the
firms. This situation  appears naturally, if there are legislative
regulations imposing a hard price cap in the market (see Correa et al.~\cite{CorreaGLNS18} and Harks et al.~\cite{HarksSV19}
and further references therein).
As reported in~\cite{CorreaGLNS18}, even different price caps for different firms is current practice in the highway market of Santiago de Chile, where 12 different operators  set tolls on different urban highways, each with a specific price cap.
While firm-specific price caps generalize the model of \JWvR,  we impose, on the other hand, more
restrictive assumptions on the  congestion 
functions, namely that they are linear with respect to the volume of customers
and inverse linear with respect to the installed capacity. 

As our first result,
we completely characterize the
structure of best response correspondences
of firms; including the possibility
of non-existence of a best response (Theorem~\ref{theo:BR1}). 
Our second main result then establishes the
existence of equilibria (Theorem~\ref{theo:existence}). 
For the proof
we use the concept of $C$-security
introduced by McLennan, Monteiro and Tourky~\cite{McLennan11}, 
which in turn resembles ideas of Reny~\cite{Reny99}.
A game is \emph{$C$-secure} at a given
strategy profile, if
each player has a pure strategy guaranteeing
a certain utility value, even if the other players
play some perturbed strategy within a (small enough)
neighbourhood, and furthermore, for each slightly perturbed strategy profile, there is a player whose perturbed strategy can in some sense be strictly separated from her securing strategies. 
Intuitively, the concept of securing strategies means that
those strategies are robust to other players' small deviations.
The result of~\cite{McLennan11} states
that a game with compact, convex strategy sets and bounded profit functions admits an equilibrium, if every non-equilibrium profile is $C$-secure.
It is important to note that the concept of  $C$-security does
not rely on quasi-concavity or continuity of profit functions. 
 With our characterization of best response
 correspondences  at hand, we show that the capacity and price competition game 
 with inelastic demand, linear congestion functions and
 price caps
fulfills the conditions of McLennan
et al.'s result 
and thus  admits pure Nash equilibria. 
As our third main result, we show that the equilibrium is essentially unique (Theorem~\ref{theo:uniqueness}). 
The general proof approach is the same as in \JWvR. However, since
our model includes price caps, we need to adjust their approach in order to work for our
model. In particular, the set of firms having positive capacity needs to be decomposed,
where the decomposition is related to the property whether the price of a firm is equal
to its cap, or strictly smaller. 
We finally study the worst case efficiency of the unique equilibrium compared to a natural benchmark, in which we relax the equilibrium conditions of the firms, but not the equilibrium conditions of the customers. We show that the unique equilibrium might be arbitrarily inefficient (Theorem~\ref{theo:quality}), by presenting a family of instances such that the
quality of the unique equilibrium gets arbitrarily bad.

\section{Model}\label{sec:model}
In a \textit{capacity and price competition game}, there is a set $N=\{1,\dots,n\}$, $n\geq 2$, of firms offering a service
to customers. Customers are represented by the continuum $[0,1]$ (each consumer
is assumed to be infinitesimally small and represented by a number in $[0,1]$) 
and we denote by $P=\{x\in \R_{\geq 0}^n\vert  \sum_{i\in N}x_i=1\}$
the standard simplex of assignments of customers to firms.\footnote{All results hold for arbitrary intervals $[0,d], d\in \R_{\geq 0}$ by a standard scaling argument} 
The effective quality of the service of firm $i$ offered to a customer
depends on two key factors: the level of congestion $\ell_i(x_i,z_i)$ and the price $p_i\geq 0$ charged
by the firm.
The congestion function $\ell_i(x_i,z_i)$ depends on the volume of customers $x_i$
and the service capacity $z_i\geq 0$ installed.
Clearly $\ell_i(x_i,z_i)$ grows with the volume of customers~$x_i$
to be served but decreases with the service
capacity $z_i$. 
If no capacity is installed, i.e. $z_i=0$, we assume infinite congestion, and for the case that $z_i>0$, we assume that congestion depends \emph{linearly} on the volume of customers and
 \emph{inverse-linearly} on the installed capacity, that is,
 \[ \ell_i(x_i,z_i)=\begin{cases} \frac{a_i x_i}{z_i}+b_i,
 & \text{ for $z_i>0$}, \\ \infty, & \text{ for $z_i=0$}, \end{cases}
\]
 where $a_i>0$ and $b_i\geq 0$ are given parameters for $i\in N$.
Note that the case $z_i=0$ can be interpreted as if firm $i$ is just opting out of the market and does not offer the service at all.
Each firm~$i$ additionally decides on a price $p_i \in [0,C_i]$ which is charged for offering
service to its customers, where $C_i >0$ is a given price cap. 
For a capacity vector $z=(z_1,\ldots,z_n)$ with $\sum_{i\in N}{z_i}>0$, i.e. the service is offered by some firms, and a price vector $p=(p_1,\ldots,p_n)$, customers
choose rationally the most attractive service in terms of the \emph{effective costs}, that is, congestion and price experienced. 
This is expressed by the \emph{Wardrop equilibrium} conditions (where we assume that congestion is measured in monetary units): 
\[ c_i(x,z,p):=\ell_i(x_i,z_i)+p_i\leq c_j(x,z,p):= \ell_j(x_j,z_j)+p_j \]
holds for all $i,j\in N$ with $x_i>0$. 
Note that for given capacities $z\neq 0$ and prices $p$, there is exactly one $x\in P$ satisfying the Wardrop equilibrium conditions (see, e.g.,~\cite{CorreaS2011}). Call this flow $x=x(z,p)$ the Wardrop flow \textit{induced by $(z,p)$}. In particular, there exists a constant $K\geq 0$ such that $c_i(x,z,p)=K$ holds for each $i\in N$ with $x_i>0$, and $c_i(x,z,p) \geq K$ holds for each $i\in N$ with $x_i=0$. 
For a Wardrop flow $x$, we call the corresponding constant $K$ the \textit{(routing) cost of~$x$}.
The profit function of a firm $i\in N$ can now be represented as
\[
\Pi_i(z,p)=\begin{cases} p_ix_i(z,p)-\gamma_iz_i, & \text{ for $\sum_{i\in N}{z_i}>0$}, \\ 0, & \text{ else}, \end{cases}
\]
where $\gamma_i>0$ is a given installation cost parameter for firm~$i\in N$.   
We assume that each firm~$i \in N$ seeks to maximize her own profit $\Pi_i$. For each firm~$i\in N$, let $S_i:=\{s_i=(z_i,p_i): 0\leq z_i, 0\leq p_i \leq C_i\}$ be her strategy set.   
A vector $s$ consisting of strategies $s_i=(z_i,p_i) \in S_i$ for all $i\in N$ is called a strategy profile, and $S:=\times_{i \in N}S_i$ denotes the set of strategy profiles. 
Usually, we will write a strategy profile $s \in S$ in the form $s=(z,p)$, where $z$ denotes the vector consisting of all capacities $z_i$ for $i\in N$, and $p$ is the vector of prices $p_i$ for $i\in N$. The profit of firm $i$ for a strategy profile $s=(z,p)$ is then defined as $\Pi_i(s):=\Pi_i(z,p)$. 
Furthermore, we write $x(s):=x(z,p)$ for the Wardrop-flow induced by $s=(z,p)$ and $K(s):=K(z,p)$ for the routing cost of $x(s)$. 
For firm $i$, denote by $s_{-i}=(z_{-i},p_{-i})\in S_{-i}:=\times_{j\in N \setminus \{i\}}{S_j}$ the vector consisting of strategies $s_j=(z_j,p_j)\in S_j$ for all $j\in N\setminus \{i\}$. 
We then write $(s_i,s_{-i})=((z_i,p_i),(z_{-i},p_{-i}))$ for the strategy profile where firm $i$ chooses $s_i=(z_i,p_i)\in S_i$, and the other firms choose $s_{-i}=(z_{-i},p_{-i})\in S_{-i}$.
Moreover, we use the simplified notation
$\Pi_i((s_i,s_{-i}))=\Pi_i(s_i,s_{-i})$ and $x((s_i,s_{-i}))=x(s_i,s_{-i})$.  
A strategy profile $s=(s_i,s_{-i})$ is  a \textit{pure Nash equilibrium (PNE)}, 
if for each firm $i\in N$:
\[
\Pi_i\klammer{s_i,s_{-i}}\geq \Pi_i\klammer{s_i',s_{-i}} \text{ for all } s_i'\in S_i.
\]
For given strategies $s_{-i}\in S_{-i}$ of the other firms, the best response correspondence of firm $i$ is
defined by 
\[ \BR_i(s_{-i}):=\arg\max\{\Pi_i\klammer{s_i,s_{-i}} \vert s_i\in S_i\}.\]
If $s_{-i}$ is clear from the context, we just write $\BR_i$ instead of $\BR_i(s_{-i})$.

We conclude this section with the following fundamental result about the continuity of the profit functions. We will use Theorem~\ref{theo:continuity2}, which completely characterizes the strategy profiles $s$ having the property that all profit functions $\Pi_i$, $i\in N$, are continuous at $s$, several times during the rest of the paper. 
\begin{theorem}\label{theo:continuity2}
Let $s=(z,p)\in S$. Then: The profit function $\Pi_i$ is continuous at $s=(z,p)$ for all $i\in N$ if and only if $z\neq 0$ or $(z,p)=(0,0)$. 
\end{theorem}

\begin{proof}
We start with the strategy profile $s=(z,p)=(0,0)$ and show that for each firm~$i$, her profit function $\Pi_i$ is continuous at $s$. 
Let $i\in N$. For $\eps>0$, define $\delta:=\min\{\frac{\eps}{2\gamma_i}, \frac{\eps}{2}\}>0$, and let $s'\in S$ with $\norm{s'}=\norm{s'-s}<\delta$ (where $\norm{\cdot}$ denotes the Euclidean norm). If $z_i'=0$, then $\Pi_i(s')=\Pi_i(s)=0$. Otherwise, 
\[
|\Pi_i(s')-\Pi_i(s)|=|x_i(s')p_i'-\gamma_iz_i'|\leq x_i(s')p_i'+\gamma_iz_i'\leq p_i'+\gamma_iz_i'\leq \norm{s'}+\gamma_i\norm{s'}<\delta(1+\gamma_i)\leq \eps
\]
holds, showing that $\Pi_i$ is continuous at $s$.   

Now consider $s=(z,p)\in S$ with $z\neq 0$. We again need to show that all profit functions $\Pi_i$, $i\in N$, are continuous at $s$. 
Since $z\neq 0$, we get that $N^+:=\{j\in N: z_j>0\}\neq \emptyset$. Furthermore, for $\delta_1>0$ sufficiently small, $N^+\subseteq \{j\in N: z_j'>0\}=:N^+(s')$ holds for all $s'=(z',p')\in S$ with $\norm{s-s'}<\delta_1$. Write $s'=(s'_1,s'_2)\in S$, where $s_1'$ denotes the strategies of the firms in $N^+$, and $s_2'$ denotes the strategies of the firms in $N\setminus N^+$. 
Now let $i\in N$. We need to show that $\Pi_i$ is continuous at $s=(z,p)$. For all $s'\in S$ with $\norm{s-s'}<\delta_1$, firm $i$'s profit is $\Pi_i(s')=x_i(s')p_i'-\gamma_iz_i'$. Thus it is sufficient to show that $x_i$ is continuous at $s=(s_1,s_2)$. 
The idea of the proof is to show that, for a slightly perturbed strategy profile $(s_1',s_2')$, the difference between $x_i(s)$ and $x_i\klammer{s_1',s_2}$, as well as the difference between $x_i\klammer{s_1',s_2}$ and $x_i\klammer{s_1',s_2'}$, is small. 
In the following, let $s'\in S$ with $\norm{s-s'}<\delta_1$. It is well known (\cite{Beckmann56}, compare also~\cite{CorreaS2011}) that $x(s')$ is the unique optimal solution of the following optimization problem $\text{Q}=\text{Q}(s')$: 
\[
(\text{Q}) \quad \min \ \ \sum_{j\in N}{\int_{0}^{x_j}{\left(\ell_j(t,z_j')+p_j'\right)\dt}} \quad \text{ s.t.} \ \ \sum_{j\in N}{x_j}=1, \ x_j\geq 0 \ \forall j\in N.
\]
Furthermore, $x_j(s')=0$ for $j\notin N^+(s')$. Therefore, the values $(x_j(s'))_{j\in N^+(s')}$ are the unique optimal solution of
\[
\quad \min \sum_{j\in N^+(s')}{\int_{0}^{x_j}{\left(\frac{a_j}{z_j'}t+b_j+p_j'\right)\dt}} \quad \text{ s.t.} \sum_{j\in N^+(s')}{x_j}=1, \ x_j\geq 0 \ \forall j\in N^+(s'),
\]
which is equivalent to
\[
\max -\sum_{j\in N^+(s')}{\left(\frac{a_j}{2z_j'}x_j^2+(b_j+p_j')x_j\right)} \quad \text{ s.t.} \sum_{j\in N^+(s')}{x_j}=1, \ x_j\geq 0 \ \forall j\in N^+(s').
\]
By Berge's theorem of the maximum~\cite{berge63}, for all $\eps>0$ there is $0<\delta_2=\delta_2(\eps)<\delta_1$ such that $\norm{x(s)-x\klammer{s_1',s_2}}<\eps$ for all $(s_1',s_2)\in  S$ with $\norm{s-(s_1',s_2)}<\delta_2$. That is, $x$ is continuous at $s$ if we only allow changes in $s_1$, but not in $s_2$. Furthermore, if $q(s')$ denotes the optimal objective function value of $\text{Q}(s')$, and if only changes in $s_1$ are allowed, $q$ is also continuous in $s$, i.e. for all $\eps>0$ there is $0<\delta_3=\delta_3(\eps)<\delta_1$ such that $|q(s)-q\klammer{s_1',s_2}|<\eps$ for all $(s_1',s_2)\in S$ with $\norm{s-(s_1',s_2)}<\delta_3$.
We now distinguish between the two cases that $z_i>0$ or $z_i=0$. 

First consider the case $z_i=0$, that is, $i\notin N^+$, and let $\eps>0$. Note that $x_i(s)=0$, thus we need to find $\delta>0$ such that $|x_i(s)-x_i(s')|=x_i(s')<\eps$ for all $s'\in S$ with $\norm{s-s'}<\delta$. To this end, define 
\[
\delta=\delta(i,\eps):=\begin{cases}
\delta_3(1), & \text{ if } q(s)+1\leq b_i\eps,\\
\min\{\delta_3(1), \frac{a_i \eps^2}{2(q(s)+1-b_i\eps)}\}, & \text{ else,}
\end{cases}
\]
and let $s'\in S$ with $\norm{s-s'}<\delta$. In particular, $|z_i-z_i'|=z_i'<\delta$. Furthermore, $q(s')\leq q\klammer{s_1',s_2}\leq q(s)+1$ holds since $\norm{s-(s_1',s_2)}\leq \norm{s'-s}<\delta\leq \delta_3(1)$. 
If $z_i'=0$, we immediately get $x_i(s')=0<\eps$. Thus assume $z_i'>0$ and assume, by contradiction, that $x_i(s')\geq \eps$. Then, by definition of $\delta$, we get the following contradiction:
\[
q(s)+1\geq q(s')\geq \frac{a_i}{2z_i'}x_i(s')^2+(b_i+p_i')x_i(s')\geq \frac{a_i}{2z_i'}\eps^2+b_i\eps > \frac{a_i}{2\delta}\eps^2+b_i\eps \geq q(s)+1.
\]
Therefore, $x_i(s')< \eps$ holds, showing that $x_i$ is continuous at $s$ if $i\notin N^+$.

Now consider the case $i\in N^+$, i.e. $z_i>0$. For $\eps>0$, we need to find $\delta>0$ such that $|x_i(s)-x_i(s')|<\eps$ for all $s'\in S$ with $\norm{s-s'}<\delta$. To this end, define 
\[
\delta:=\min\{\min\{\delta(j,\frac{\eps}{2n}): j\notin N^+\}, \delta_2(\frac{\eps}{2})\}
\]
and let $s'\in S$ with $\norm{s-s'}<\delta$. In particular, $\norm{s-(s_1',s_2)}<\delta\leq\delta_2(\frac{\eps}{2})$, thus $|x_i(s)-x_i\klammer{s_1',s_2}|<\eps/2$. 
Furthermore, since $\delta\leq\delta(j,\frac{\eps}{2n})$, we get $x_j(s')\leq \frac{\eps}{2n}$ for all $j\notin N^+$. If $x_j(s')=0$ for all $j\notin N^+$, we get $x_i(s')=x_i\klammer{s_1',s_2}$ and thus 
$|x_i(s)-x_i(s')|=|x_i(s)-x_i\klammer{s_1',s_2}|<\eps/2<\eps$, as desired. Otherwise, there is $j\notin N^+$ with $0<x_j(s')\leq \frac{\eps}{2n}$. In particular, $z_j'>0$. We now use a result about the sensitivity of Wardrop flows due to Englert et al.~\cite[Theorem 2]{englert08}. 
They show that if the customers are not able to choose firm~$j$'s service anymore (we say that firm~$j$ is deleted from the customers' game), the resulting change in the Wardrop flow can be bounded by the flow that $j$ received. 
More formally, if $x\in[0,1]^n$ is the Wardrop flow for the game with firms $N$, and $x'\in [0,1]^{n-1}$ is the Wardrop flow if firm $j$ is deleted from the customers' game, then $|x_k-x_k'|<x_j$ for all $k\in N\setminus \{j\}$. 
Obviously, changing firm~$j$'s capacity from $z_j'>0$ to $z_j=0$ 
has the same effect on the Wardrop flow as deleting firm $j$. 
Therefore, if we change, one after another, the capacities of all firms $j\notin N^+$ having $z_j'>0$ to $z_j=0$, we get (note that the flow values for $j\notin N^+$ are always upper-bounded by $\eps/(2n)$ due to our choice of $\delta$ and the analysis of the case $z_i=0$):
\[
|x_i(s')-x_i\klammer{s_1',s_2}|\leq (n-1)\frac{\eps}{2n}<\frac{\eps}{2}.
\] 
Using this, we now get the desired inequality:
\[
|x_i(s)-x_i(s')|\leq |x_i(s')-x_i\klammer{s_1',s_2}|+|x_i(s)-x_i\klammer{s_1',s_2}|<\frac{\eps}{2}+\frac{\eps}{2}=\eps.
\]
Altogether we showed that all profit functions $\Pi_i$ are continuous at $s=(z,p)$ if $z\neq 0$.

To complete the proof, it remains to show that if all profit functions $\Pi_i$, $i \in N$, are continuous at $s=(z,p)$, then $z\neq 0$ or $(z,p)=(0,0)$ holds. We show this by contraposititon, thus assume that $s=(z,p)$ fulfills $z=0$ and $p\neq 0$. We need to show that there is a firm $i$ such that $\Pi_i$ is not continuous at $s$. 
To this end, let $i\in N$ with $p_i>0$. Define the sequence of strategy profiles~$s^n$ by $(z_j^n,p_j^n):=(z_j,p_j)$ for all $j\neq i$, and $(z_i^n,p_i^n):=(\frac{1}{n},p_i)$. Obviously, $s^n \rightarrow s$ for $n \rightarrow \infty$. But for the profits, we get
\[
\Pi_i(s^n)=x_i(s^n)p_i^n-\gamma_iz_i^n=p_i-\frac{\gamma_i}{n} \rightarrow_{n\rightarrow \infty}p_i>0.
\]
Since $\Pi_i(s)=0$, this shows that $\Pi_i$ is not continuous at $s$. 
\end{proof}

\section{Characterization of Best Responses}\label{section:BR}
The aim of this section is to derive a complete characterization of the best-response correspondences of the firms. We will make use of this characterization in all our main results, i.e. existence, uniqueness and quality of PNE. 
Given a firm $i\in N$ and fixed strategies $s_{-i}=(z_{-i},p_{-i})\in S_{-i}$ for the other firms, we characterize the set $\BR_i=\BR_i(s_{-i})$ of best responses of firm $i$ to $s_{-i}$. 
To this end, we distinguish between the two cases that $z_{-i}=0$ (Subsection~\ref{subsec:1}) and $z_{-i}\neq 0$ (Subsection~\ref{subsec:2}). Subsection~\ref{subsec:3} then contains the derived complete characterization. 
In Subsection~\ref{subsec:4}, we discuss how our results about the best responses influence the applicability of Kakutani's fixed point theorem. 
\subsection{The Case $z_{-i}=0$}\label{subsec:1}
In this subsection, assume that the strategies $s_{-i}=(z_{-i},p_{-i})$ of the other firms fulfill $z_{-i}=0$. Under this assumption, firm $i$ does not have a best response to $s_{-i}$: 
\begin{lemma}\label{lemma:BR0}
If $z_{-i}=0$, then $\BR_i(z_{-i},p_{-i})=\emptyset$.
\end{lemma}
\begin{proof}
Whenever firm $i$ chooses a strategy $(z_i,p_i)$ with $z_i>0$, then $x_i=1$ holds for the induced Wardrop-flow $x$, thus firm $i$'s profit is $p_i-\gamma_iz_i$. On the other hand, any strategy $(z_i,p_i)$ with $z_i=0$ yields a profit of 0. Thus, firm $i$'s profit depends solely on her own strategy $(z_i,p_i)$, and can be stated as follows:
\[
\Pi_i(z_i,p_i):=
\begin{cases}
p_i-\gamma_iz_i, & \text{ for $z_i>0$,} \\
0, & \text{ for $z_i=0$.}
\end{cases}
\]
Obviously, $\Pi_i(z_i,p_i)<C_i$ holds for each $(z_i,p_i)\in S_i$, i.e. for $z_i\geq 0$ and $0\leq p_i\leq C_i$. On the other hand, by $(z_i,p_i)=(\eps,C_i)$ for $\eps>0$, firm $i$ gets a profit of $C_i-\gamma_i\cdot \eps$ arbitrarily near to $C_i$, that is, $\sup\{\Pi_i(z_i,p_i): (z_i,p_i)\in S_i\}=C_i$. This shows $\BR_i(z_{-i},p_{-i})=\emptyset$. 
\end{proof}

\subsection{The Case $z_{-i}\neq 0$}\label{subsec:2}
In this subsection, assume that the strategies $s_{-i}=(z_{-i},p_{-i})$ of the other firms fulfill $z_{-i}\neq 0$. 
For a strategy $s_i=(z_i,p_i)$ of firm $i$, write $\Pi_i(z_i,p_i):=\Pi_i\klammer{s_i,s_{-i}}$ for firm $i$'s profit function, $x(z_i,p_i):=x\klammer{s_i,s_{-i}}$ for the Wardrop-flow induced by $(s_i,s_{-i})$ and $K(z_i,p_i):=K\klammer{s_i,s_{-i}}$ for the corresponding routing cost.

For $(z_i,p_i)\in S_i$, firm $i$'s profit is 
$\Pi_i(z_i,p_i)=x_i(z_i,p_i)p_i-\gamma_iz_i$.
It is clear that each strategy $(z_i,p_i)$ with $z_i=0$ yields $x_i(z_i,p_i)=0$, and thus $\Pi_i(z_i,p_i)=0$. On the other hand, each strategy $(z_i,p_i)$ with $z_i>\frac{C_i}{\gamma_i}$ yields negative profit since
\[
\Pi_i(z_i,p_i)=x_i(z_i,p_i)p_i -\gamma_iz_i < 1\cdot C_i -\gamma_i \cdot C_i/\gamma_i=0.
\]
Therefore, each best response $(z_i,p_i)$ fulfills $z_i\leq \frac{C_i}{\gamma_i}$ since it yields nonnegative profit.\label{bound_capacities}
Thus, $\BR_i$ can be described as the set of optimal solutions of the following optimization problem \hyperlink{P_i}{$(\text{P}_i)$}:
\[
 \hypertarget{P_i}{(\text{P}_i)} \qquad \max \ \Pi_{i}(z_i,p_i) 
\ \text{ subject to} \ z_i \in [0,\frac{C_i}{\gamma_i}], \ p_i\in[0, C_i].
\]
\hyperlink{P_i}{$(\text{P}_i)$} has an optimal solution, since the theorem of Weierstrass can be applied: The feasible set of~\hyperlink{P_i}{$(\text{P}_i)$} is compact and nonempty, and $\Pi_i$ is continuous at $(z_i,p_i)$ for all feasible $(z_i,p_i)$ (see Theorem~\ref{theo:continuity2}).
Since $\BR_i$ can be described as the set of optimal solutions of \hyperlink{P_i}{$(\text{P}_i)$}, we get $\BR_i\neq \emptyset$. 

Note that \hyperlink{P_i}{$(\text{P}_i)$} is a \textit{bilevel} optimization problem (since $x(z_i,p_i)$ can be described as the optimal solution of a minimization problem~\cite{Beckmann56}), and these problems are known to be notoriously hard to solve. 
The characterization of $\BR_i$ that we derive here has the advantage that it only uses \textit{ordinary} optimization problems, namely the following two (1-dimensional) optimization problems in the variable $K\in \R$, 

\begin{minipage}[t]{0.45\textwidth}
\begin{align*}
 \hypertarget{P_1}{(\text{P}_i^1)}\\ \; \max\; & f_i^1(K):=   \overline x_i(K)\cdot(K-b_i-2\sqrt{a_i\gamma_i})\\[0.3em]
\quad \text{ s.t. } & \quad 2\sqrt{a_i\gamma_i}+b_i \leq K \\%[0.3em]
& \quad K \leq \sqrt{a_i\gamma_i}+b_i+C_i\\%[0.3em]
%& \quad K<K_i^{\max}
& \quad \overline x_i(K)>0,
\end{align*} 
\end{minipage}
\begin{minipage}[t]{0.55\textwidth}
\begin{align*}
 \hypertarget{P_2}{(\text{P}_i^2)}\\ \; \max \; &  f_i^2(K):=  \overline x_i(K)\cdot(C_i-\frac{a_i\gamma_i}{K-b_i-C_i})\\
\quad \text{ s.t. } & \quad \sqrt{a_i\gamma_i}+b_i+C_i < K\\
& \quad \frac{a_i\gamma_i}{C_i}+b_i+C_i\leq K \\
%& \quad K<K_i^{\max},
& \quad \overline x_i(K)>0,
\end{align*}
\end{minipage}
where \label{page_K_i_max}
\[
\overline x_i(K):=1-\sum_{j \in N(K)}{\frac{(K-b_j-p_j)z_j}{a_j}} 
\]
with $N^+:=\{j\in N\setminus \{i\}:z_j>0\}$ and $N(K):=\{j\in N^+: b_j+p_j<K\}$. 
The objective functions of  \hyperlink{P_1}{$(\text{P}_i^1)$} and  \hyperlink{P_2}{$(\text{P}_i^2)$} are denoted by $f_i^1$ and $f_i^2$, respectively. 

Note that $\overline x_i: \R \rightarrow \R$ is a continuous function which is equal to 1 for $K\leq \min\{b_j+p_j: j\in N^+\}$, and strictly decreasing for $K\geq \min\{b_j+p_j: j\in N^+\}\geq 0$ with $\lim_{K\rightarrow \infty}\overline x_i(K)=-\infty$. Therefore, there is a unique constant $K_i^{\max}>0$ with the property $\overline x_i(K_i^{\max})=0$. Obviously, $x_i(K)>0$ if and only if $K<K_i^{\max}$. 
Furthermore, the function $\overline x_i$ is closely related to Wardrop-flows, as described in the following lemma. 
\begin{lemma}\label{lemma:BR_x_i} \
\begin{enumerate}[1.]
	\item\label{lemma:BR_x_i_1} If $(z_i,p_i)\in S_i$ with $x_i(z_i,p_i)>0$, then $\overline x_i(K)=x_i(z_i,p_i)$ for $K:=K(z_i,p_i)$. 
	\item\label{lemma:BR_x_i_2} If $K\geq 0$ with $\overline x_i(K)>0$, and $(z_i,p_i)\in S_i$ fulfills $z_i>0$ and $\frac{a_i}{z_i}\overline x_i(K)+b_i+p_i=K$, then $x_i(z_i,p_i)=\overline x_i(K)$ and $K(z_i,p_i)=K$. 
\end{enumerate}
\end{lemma}

\begin{proof}
We start with statement~\ref{lemma:BR_x_i_1} of the lemma, so let $(z_i,p_i)\in S_i$ with $x_i(z_i,p_i)>0$. By definition of $K(z_i,p_i)=:K$, we get $\ell_j(x_j(z_i,p_i), z_j)+p_j=K$ for all $j\in N$ with $x_j(z_i,p_i)>0$, and $\ell_j(x_j(z_i,p_i), z_j)+p_j\geq K$ for all $j\in N$ with $x_j(z_i,p_i)=0$. 
Since 
\[
\ell_j(x_j(z_i,p_i), z_j)+p_j=\begin{cases}
\frac{a_j}{z_j}x_j(z_i,p_i)+b_j+p_j, & \text{ for $j \in N$ with $z_j>0$,} \\
\infty, & \text{ for $j\in N$ with $z_j=0$,}
\end{cases}
\] 
we get that $\{j \in N: x_j(z_i,p_i)>0\}=\{j \in N^+: b_j+p_j<K\}\cup \{i\}=N(K)\cup \{i\}$. Therefore, for each $j\in N(K)$, we get $\frac{a_j}{z_j}x_j(z_i,p_i)+b_j+p_j=K$, which is equivalent to $x_j(z_i,p_i)=\frac{(K-b_j-p_j)z_j}{a_j}$. Using $\sum_{j\in N}{x_j(z_i,p_i)}=1$ yields
\[
x_i(z_i,p_i)=1-\sum_{j\in N \setminus \{i\}}{x_j(z_i,p_i)}=1-\sum_{j \in N(K)}{\frac{(K-b_j-p_j)z_j}{a_j}}=\overline x_i(K),
\]
as desired. 

Now we turn to statement~\ref{lemma:BR_x_i_2} of the lemma, so let $K\geq 0$ with $\overline x_i(K)>0$ and let $(z_i,p_i)\in S_i$ be a strategy with $z_i>0$ and $\frac{a_i}{z_i}\overline x_i(K)+b_i+p_i=K$. 
Consider the vector $x\in [0,1]^n$ defined by
\[
x_j:=\begin{cases}\overline x_i(K), & j=i, \\ \frac{(K-b_j-p_j)z_j}{a_j}, & j\in N^+ \text{ with } b_j+p_j <K, \\
0, & j\in N^+ \text{ with } b_j+p_j \geq K \text{ or } j\in N\setminus (N^+\cup \{i\}).
 \end{cases}
\]
It is clear that $x_j>0$ holds for all $j\in N^+$ with $b_j+p_j<K$, and $x_i=\overline x_i(K)>0$. Furthermore, the definition of $\overline x_i(K)$ yields $\sum_{j\in N}{x_j}=1$. We now show that $x$ fulfills the Wardrop equilibrium conditions. The uniqueness of the Wardrop-flow then implies $x=x(z_i,p_i)$, and $K(z_i,p_i)=K$ follows from $x_i(z_i,p_i)=\overline x_i(K)>0$ and $K(z_i,p_i)=\frac{a_i}{z_i}x_i(z_i,p_i)+b_i+p_i=\frac{a_i}{z_i}\overline x_i(K)+b_i+p_i=K$, completing the proof. 
For the Wardrop equilibrium conditions, consider the effective costs of the firms:
\[
c_j(x,z,p)
=\begin{cases}
 \frac{a_i}{z_i} \overline x_i(K) +b_i+p_i=K, & j=i, \\
 \frac{a_j}{z_j}\cdot \frac{(K-b_j-p_j)z_j}{a_j}+b_j+p_j=K, & j\in N^+ \text{ with } b_j+p_j<K, \\
 b_j+p_j \geq K, & j\in N^+ \text{ with } b_j+p_j\geq K, \\
\infty >K, & j \in N \setminus (N^+\cup \{i\}).
 \end{cases}
\]
It is now clear that $x$ fulfills the Wardrop equilibrium conditions.
\end{proof}

In the following lemmata, we analyze the connection between \hyperlink{P_1}{$(\text{P}_i^1)$} and  \hyperlink{P_2}{$(\text{P}_i^2)$} and the optimal solutions of \hyperlink{P_i}{$(\text{P}_i)$}. 
\begin{lemma}\label{lemma:BR1} \
\begin{enumerate}[1.]
	\item\label{lemma:BR1_1} If $K$ is feasible for problem \hyperlink{P_1}{$(\text{P}_i^1)$}, the tuple $(z_i,p_i)$ defined by
	\[
	z_i:=\sqrt{a_i/\gamma_i}\cdot \overline x_i(K), \quad p_i:=K-\sqrt{a_i\gamma_i}-b_i
	\]
	is feasible for \hyperlink{P_i}{$(\text{P}_i)$}, and fulfills $z_i>0$ and $\Pi_i(z_i,p_i)=f_i^1(K)$. 
	\item\label{lemma:BR1_2} If $K$ is feasible for problem \hyperlink{P_2}{$(\text{P}_i^2)$}, the tuple $(z_i,p_i)$ defined by
	\[
	z_i:=\frac{a_i \cdot \overline x_i(K)}{K-b_i-C_i}, \quad p_i:=C_i
	\]
	is feasible for \hyperlink{P_i}{$(\text{P}_i)$}, and fulfills $z_i>0$ and $\Pi_i(z_i,p_i)=f_i^2(K)$. 
\end{enumerate}
\end{lemma}

\begin{proof}
We start with statement~\ref{lemma:BR1_1} of the lemma, thus assume that $K$ is feasible for problem \hyperlink{P_1}{$(\text{P}_i^1)$}.  
Let $z_i:=\sqrt{a_i/\gamma_i}\cdot \overline x_i(K)$ and $p_i:=K-\sqrt{a_i\gamma_i}-b_i$ as stated in~\ref{lemma:BR1_1}. 
The feasibility of $K$ for \hyperlink{P_1}{$(\text{P}_i^1)$} yields $\overline x_i(K)>0$ and $2\sqrt{a_i\gamma_i}+b_i\leq K\leq \sqrt{a_i\gamma_i}+b_i+C_i$. 
From this we conclude $z_i>0$ and $0<p_i=K-\sqrt{a_i\gamma_i}-b_i\leq C_i$, thus 
$(z_i,p_i)\in S_i$. 
Furthermore, 
\[
\frac{a_i}{z_i}\cdot \overline x_i(K)+b_i+p_i=
\frac{a_i}{\sqrt{a_i/\gamma_i}\cdot \overline x_i(K)}\cdot \overline x_i(K) +b_i+K-\sqrt{a_i\gamma_i}-b_i
=K
\]
holds, thus we get $x_i(z_i,p_i)=\overline x_i(K)$ from statement~\ref{lemma:BR_x_i_2} of Lemma~\ref{lemma:BR_x_i}. 
Using this, we can now show that firm~$i$'s profit for $(z_i,p_i)$ equals the objective function value of $K$ for \hyperlink{P_1}{$(\text{P}_i^1)$}:
\begin{align*} 
\Pi_i(z_i,p_i)&=p_i x_i(z_i,p_i)-\gamma_i z_i 
=(K-\sqrt{a_i\gamma_i}-b_i)\cdot \overline x_i(K)-\gamma_i\cdot \sqrt{a_i/\gamma_i}\cdot \overline x_i(K) \\
&= \overline x_i(K)\cdot(K-2\sqrt{a_i\gamma_i}-b_i)=f_i^1(K)
\end{align*}
Note that the feasibility of $K$ for \hyperlink{P_1}{$(\text{P}_i^1)$} yields $f_i^1(K)\geq 0$. 
It remains to show that $(z_i,p_i)$ is feasible for \hyperlink{P_i}{$(\text{P}_i)$}. 
We already know that $z_i>0$ and $0<p_i\leq C_i$ holds. The remaining inequality $z_i\leq \frac{C_i}{\gamma_i}$ follows from the nonnegativity of $\Pi_i(z_i,p_i)=f_i^1(K)\geq 0$ and the fact that any strategy with $z_i>\frac{C_i}{\gamma_i}$ yields negative profit for firm~$i$. 

Now turn to statement~\ref{lemma:BR1_2} of the lemma. 
Assume that $K$ is feasible for~\hyperlink{P_2}{$(\text{P}_i^2)$}, and let $z_i:=\frac{a_i \cdot \overline x_i(K)}{K-b_i-C_i}$ and $p_i:=C_i$. 
The feasibility of $K$ for \hyperlink{P_2}{$(\text{P}_i^2)$} implies $K>\sqrt{a_i\gamma_i}+b_i+C_i>b_i+C_i$ and $\overline x_i(K)>0$, thus $z_i>0$ holds and this yields $(z_i,p_i)\in S_i$.  Furthermore, 
\[
\frac{a_i}{z_i}\cdot \overline x_i(K)+b_i+p_i=
\frac{a_i}{\frac{a_i \cdot \overline x_i(K)}{K-b_i-C_i}}\cdot \overline x_i(K) +b_i+C_i
=K
\]
holds, thus we get $x_i(z_i,p_i)=\overline x_i(K)$ from statement~\ref{lemma:BR_x_i_2} of Lemma~\ref{lemma:BR_x_i}.
The profit of firm $i$ thus becomes
\begin{align*}
\Pi_i(z_i,p_i)&=p_ix_i(z_i,p_i)-\gamma_iz_i 
=C_i \cdot \overline x_i(K)-\gamma_i \cdot \frac{a_i \overline x_i(K)}{K-b_i-C_i} \\
&= \overline x_i(K) \cdot (C_i-\frac{a_i\gamma_i}{K-b_i-C_i})=f_i^2(K).
\end{align*}
Note that $f_i^2(K)\geq 0$ holds due to the feasibility of $K$ for \hyperlink{P_2}{$(\text{P}_i^2)$}.
 As in the proof of statement~\ref{lemma:BR1_1} of the lemma, this implies $z_i\leq C_i/\gamma_i$, and thus $(z_i,p_i)$ is feasible for \hyperlink{P_i}{$(\text{P}_i)$}, which completes the proof. 
\end{proof}
In particular, Lemma~\ref{lemma:BR1} shows that any optimal solution of \hyperlink{P_1}{$(\text{P}_i^1)$} or \hyperlink{P_2}{$(\text{P}_i^2)$} yields a feasible strategy for \hyperlink{P_i}{$(\text{P}_i)$} with the same objective fuction value. The next lemma shows that for certain optimal solutions of \hyperlink{P_i}{$(\text{P}_i)$}, the converse is also true. 
\begin{lemma}\label{lemma:BR2} 
Let $(z_i^*,p_i^*)$ be an optimal solution of \hyperlink{P_i}{$(\text{P}_i)$} and $K^*:=K(z_i^*,p_i^*)$. 
If $z_i^*>0$, then exactly one of the following two cases holds:
\begin{enumerate}[1.]
	\item\label{lemma:BR2_1} $(z_i^*,p_i^*)=\left(\sqrt{a_i/\gamma_i}\cdot \overline x_i(K^*), K^*-\sqrt{a_i\gamma_i}-b_i\right)$; $K^*$ is optimal for \hyperlink{P_1}{$(\text{P}_i^1)$} with $f_i^1(K^*)=\Pi_i(z_i^*,p_i^*)$.
	\item\label{lemma:BR2_2} $(z_i^*,p_i^*)=\left(\frac{a_i\cdot \overline x_i(K^*)}{K^*-b_i-C_i},C_i\right)$; $K^*$ is optimal for \hyperlink{P_2}{$(\text{P}_i^2)$} with $f_i^2(K^*)=\Pi_i(z_i^*,p_i^*)$.
\end{enumerate}  
\end{lemma}

\begin{proof}
Let $(z_i^*,p_i^*)$ with $z_i^*>0$ and $K^*$ as in the lemma statement, and 
define $x^*:=x(z_i^*,p_i^*)$. Since $\Pi_i(z_i^*,p_i^*)\geq 0$ 
and $z_i^*>0$ holds, $0\leq \Pi_i(z_i^*,p_i^*)=p_i^*x_i^*-\gamma_iz_i^*<p_i^*x_i^*$ follows, which implies $x_i^*>0$ and $p_i^*>0$. 
Therefore $K^*=\frac{a_i}{z_i^*}\cdot x_i^*+b_i+p_i^*$ holds and $(z_i^*,p_i^*)$ is an optimal solution for the following optimization problem (with variables $z_i$ and $p_i$):
\begin{align*}
 \text{ (P) } \qquad\qquad \max & \quad p_i x_i^*-\gamma_iz_i \\
\qquad\qquad \text{ subject to } & \quad \frac{a_i}{z_i}x_i^*+b_i+p_i=K^* \\
& \quad 0< z_i, \ 0< p_i\leq C_i.
\end{align*} 
Note that the optimal solutions of (P) correspond to all best responses for firm $i$ such that $x^*$ remains the Wardrop flow.
Reformulating the equality constraint in (P) yields 
$p_i=K^*-\frac{a_i}{z_i}x_i^*-b_i$.
The constraints $0< p_i\leq C_i$ then become (note that $K^*>b_i+p_i^*\geq b_i$ holds) 
\[
0<K^*-\frac{a_i}{z_i}x_i^*-b_i \Leftrightarrow \frac{1}{z_i}<\frac{K^*-b_i}{a_ix_i^*} \Leftrightarrow z_i>\frac{a_ix_i^*}{K^*-b_i}
\]
and
\[
K^*-\frac{a_i}{z_i}x_i^*-b_i\leq C_i \Leftrightarrow \frac{1}{z_i}\geq\frac{K^*-b_i-C_i}{a_ix_i^*}.
\]
Thus (P) is equivalent to the following problem (with variable $z_i$; note that $\frac{a_ix_i^*}{K^*-b_i}>0$ holds):
\begin{align*}
 \text{ (P$'$) } \qquad\qquad \max & \quad  (K^*-\frac{a_i}{z_i}x_i^*-b_i)\cdot x_i^*-\gamma_iz_i \\
\qquad\qquad \text{ subject to } & \quad \frac{a_ix_i^*}{K^*-b_i} < z_i, \ \frac{K^*-b_i-C_i}{a_ix_i^*}\leq\frac{1}{z_i}.
\end{align*} 
Let $f$ be the objective function of (P$'$) and consider the derivative $f'(z_i)=\frac{a_i(x_i^*)^2}{z_i^2}-\gamma_i$. 
We get that $f$ is strictly increasing for $0<z_i<\sqrt{a_i/\gamma_i} \cdot x_i^*$ and strictly decreasing for $z_i>\sqrt{a_i/ \gamma_i} \cdot x_i^*$. 
We now distinguish between the two cases that $z_i=\sqrt{a_i/\gamma_i}\cdot x_i^*$ is feasible for (P$'$), or not. As we will see, the former case leads to statement~\ref{lemma:BR2_1} of the lemma, and the latter case to statement~\ref{lemma:BR2_2}. 
Note that in either case, $\overline x_i(K^*)=x_i^*$ holds (by statement~\ref{lemma:BR_x_i_1} of   Lemma~\ref{lemma:BR_x_i}).

If $z_i=\sqrt{a_i/\gamma_i}\cdot x_i^*$ is feasible for (P$'$), it is the unique optimal solution of (P$'$). But since $z_i^*$ is also optimal for (P$'$), we get
\[
z_i^*=\sqrt{a_i/\gamma_i}\cdot \overline x_i(K^*) \text{ and } p_i^*=K^*-\frac{a_i}{z_i^*}\cdot \overline x_i(K^*)-b_i=K^*-\sqrt{a_i\gamma_i}-b_i.
\]
For the profit of firm $i$, we get 
\begin{align*}
\Pi_i(z_i^*,p_i^*)&=p_i^*x_i^*-\gamma_iz_i^*=(K^*-\sqrt{a_i\gamma_i}-b_i)\cdot \overline x_i(K^*)-\gamma_i\cdot\sqrt{\frac{a_i}{\gamma_i}} \cdot \overline x_i(K^*)\\
&= \overline x_i(K^*)\cdot(K^*-2\sqrt{a_i\gamma_i}-b_i)=f_i^1(K^*).
\end{align*}
It remains to show that $K^*$ is optimal for \hyperlink{P_1}{$(\text{P}_i^1)$}. 
For feasibility, we need $2\sqrt{a_i\gamma_i}+b_i\leq K^*\leq \sqrt{a_i\gamma_i}+b_i+C_i$ and $\overline x_i(K^*)>0$. 
The last inequality follows directly from $\overline x_i(K^*)=x_i^*>0$. Using this and $\overline x_i(K^*)\cdot(K^*-2\sqrt{a_i\gamma_i}-b_i)=\Pi_i(z_i^*,p_i^*)\geq 0$ yields $K^*\geq 2\sqrt{a_i\gamma_i}+b_i$. Finally, the feasibility of $z_i^*=\sqrt{a_i/\gamma_i}\cdot x_i^*$ for (P$'$) implies
\[
\frac{K^*-b_i-C_i}{a_ix_i^*} \leq \frac{\sqrt{\gamma_i}}{\sqrt{a_i}x_i^*} \Leftrightarrow K^*\leq \sqrt{a_i\gamma_i}+b_i+C_i.
\]
Therefore, $K^*$ is feasible for \hyperlink{P_1}{$(\text{P}_i^1)$}. The optimality follows from Lemma~\ref{lemma:BR1} and the optimality of $(z_i^*,p_i^*)$ for \hyperlink{P_i}{$(\text{P}_i)$}.

Now turn to the case that $z_i=\sqrt{a_i/\gamma_i}\cdot x_i^*$ is not feasible for (P$'$). We show that statement~\ref{lemma:BR2_2} of the lemma holds. 
Since (P$'$) has an optimal solution (namely $z_i^*$), we get that $0<\frac{a_ix_i^*}{K^*-b_i-C_i}<\sqrt{a_i/\gamma_i}\cdot x_i^*$ holds, and therefore $z_i=\frac{a_ix_i^*}{K^*-b_i-C_i}$ is the unique optimal solution for (P$'$). 
This shows 
\[
z_i^*=\frac{a_i\cdot \overline x_i(K^*)}{K^*-b_i-C_i} \text{ and } p_i^*=K^*-\frac{a_i}{z_i^*}\cdot \overline x_i(K^*)-b_i=C_i.
\]
The profit of firm $i$ becomes 
\[
\Pi_i(z_i^*,p_i^*)=C_i\cdot \overline x_i(K^*)-\gamma_i\cdot \frac{a_i\overline x_i(K^*)}{K^*-b_i-C_i}=\overline x_i(K^*)\cdot(C_i-\frac{a_i\gamma_i}{K^*-b_i-C_i})=f_i^2(K^*).
\]
Since $\overline x_i(K^*)=x_i^*>0$ and the profit is nonnegative, 
\[
C_i-\frac{a_i\gamma_i}{K^*-b_i-C_i} \geq 0 \Leftrightarrow K^* \geq \frac{a_i\gamma_i}{C_i}+b_i+C_i 
\]
holds.  Finally,
\[
z_i^*=\frac{a_ix_i^*}{K^*-b_i-C_i}<\sqrt{\frac{a_i}{\gamma_i}}x_i^* \Leftrightarrow K^*>\sqrt{a_i\gamma_i}+b_i+C_i,
\]
which completes the proof since we showed that $K^*$ is a feasible solution of problem \hyperlink{P_2}{$(\text{P}_i^2)$} (optimality follows from Lemma~\ref{lemma:BR1} and the optimality of $(z_i^*,p_i^*)$ for \hyperlink{P_i}{$(\text{P}_i)$}).
\end{proof}
In the next lemma, we analyze the existence of optimal solutions for the problems \hyperlink{P_1}{$(\text{P}_i^1)$} and \hyperlink{P_2}{$(\text{P}_i^2)$}, as well as properties of such solutions. 
\begin{lemma}\label{lemma:BR3} \
\begin{enumerate}[1.]
	\item\label{lemma:BR3_1} If \hyperlink{P_1}{$(\text{P}_i^1)$} is feasible, it has a unique optimal solution.
	\item\label{lemma:BR3_2} Assume that \hyperlink{P_2}{$(\text{P}_i^2)$} is feasible. 
	\begin{itemize}
		\item If $C_i\leq\sqrt{a_i\gamma_i}$, then \hyperlink{P_2}{$(\text{P}_i^2)$} has a unique optimal solution. 
		\item If $C_i> \sqrt{a_i\gamma_i}$, then \hyperlink{P_2}{$(\text{P}_i^2)$} has at most one optimal solution. 
		\item If $K^*_2$ is optimal for \hyperlink{P_2}{$(\text{P}_i^2)$}, then $f_i^2(K^*_2)>0$.
	\end{itemize}
	\item\label{lemma:BR3_3} If $K^*_1$ is optimal for \hyperlink{P_1}{$(\text{P}_i^1)$} and $K^*_2$ is optimal for \hyperlink{P_2}{$(\text{P}_i^2)$}, then $f_i^1(K^*_1)< f_i^2(K^*_2)$.
\end{enumerate}
\end{lemma}
\begin{proof}
We start with statement~\ref{lemma:BR3_1} of the lemma, so assume that \hyperlink{P_1}{$(\text{P}_i^1)$} is feasible. 
Note that the feasible set $I_1$ of \hyperlink{P_1}{$(\text{P}_i^1)$} either is of the form $I_1=[2\sqrt{a_i\gamma_i}+b_i,\sqrt{a_i\gamma_i}+b_i+C_i]$, or $I_1=[2\sqrt{a_i\gamma_i}+b_i, K_i^{\max})$, depending on whether $\sqrt{a_i\gamma_i}+b_i+C_i< K_i^{\max}$ holds or not. 
Furthermore note that the objective function 
\[
f_i^1(K)=\overline x_i(K)\cdot(K-b_i-2\sqrt{a_i\gamma_i})=\left(1-\sum_{j \in N^+: b_j+p_j<K}{\frac{(K-b_j-p_j)z_j}{a_j}}\right)\cdot(K-b_i-2\sqrt{a_i\gamma_i})
\]
of \hyperlink{P_1}{$(\text{P}_i^1)$} is continuous (over $\R$). From this, we can conclude that \hyperlink{P_1}{$(\text{P}_i^1)$} has \textit{at least one} optimal solution: 
For $I_1=[2\sqrt{a_i\gamma_i}+b_i,\sqrt{a_i\gamma_i}+b_i+C_i]$, this follows directly from the theorem of Weierstrass ($f_i^1$ is continuous and $I_1$ is nonempty and compact). 
For $I_1=[2\sqrt{a_i\gamma_i}+b_i, K_i^{\max})$, the theorem of Weierstrass yields that $f_i^1$ attains its maximum over the closure of $I_1$, that is, over $[2\sqrt{a_i\gamma_i}+b_i, K_i^{\max}]$. But since $f_i^1(K_i^{\max})=0$ $(=f_i^1(2\sqrt{a_i\gamma_i}+b_i))$, and any $K\in (2\sqrt{a_i\gamma_i}+b_i, K_i^{\max})$ fulfills $f_i^1(K)>0$, the maximum is not attained for $K=K_i^{\max}$, and we conclude that $f_i^1$ also attains its maximum over $I_1$. 
Thus \hyperlink{P_1}{$(\text{P}_i^1)$} has an optimal solution for both cases. 
To complete the proof of statement~\ref{lemma:BR3_1} of the lemma, it remains to show that there is also \textit{at most one} optimal solution. 
We prove this by showing the following monotonicity behaviour of $f_i^1$ over $I_1$: Either $f_i^1$ is strictly increasing over $I_1$, or strictly decreasing over $I_1$, or strictly increasing up to a unique point, and strictly decreasing afterwards. 
In all three cases, we obviously get the desired statement, namely that  \hyperlink{P_1}{$(\text{P}_i^1)$} has at most one optimal solution. 
To prove the described monotonicity behaviour, we distinguish between three cases according to the value of $\min\{b_j+p_j: j \in N^+\}$. 
The first case is $\sqrt{a_i\gamma_i}+b_i+C_i\leq \min\{b_j+p_j: j \in N^+\}$, which implies $N(K)=\{j \in N^+: b_j+p_j<K\}=\emptyset$ and $\overline x_i(K)=1$ for all $K\in I_1$. We conclude that $f_i^1(K)=K-b_i-2\sqrt{a_i\gamma_i}$ is strictly increasing over $I_1$ (in particular, $f_i^1$ reaches its unique maximum over $I_1$ at $K=\sqrt{a_i\gamma_i}+b_i+C_i$). 
Next, consider the case that $\min\{b_j+p_j: j \in N^+\}<2\sqrt{a_i\gamma_i}+b_i$. This implies that $N(K)\neq \emptyset$ for all $K\in I_1$. 
Note that $f_i^1$ is twice differentiable on any open interval where $N(K)$ is constant, and the first and second derivative of $f_i^1$ then are
\begin{align*}
(f_i^1)'(K) &=(-\sum_{j \in N(K)}{\frac{z_j}{a_j}})\cdot (K-b_i-2\sqrt{a_i\gamma_i}) + 1-\sum_{j \in N(K)}{\frac{(K-b_j-p_j)z_j}{a_j}}\\
&= 1-\sum_{j \in N(K)}{\frac{(2K-b_j-p_j-b_i-2\sqrt{a_i\gamma_i})z_j}{a_j}}
\end{align*}
and 
\begin{align*}
(f_i^1)''(K)&=-2\cdot \sum_{j \in N(K)}{\frac{z_j}{a_j}}.
\end{align*}
Since $N(K)\neq \emptyset$ for all $K\in I_1$, we conclude that for all $K\in I_1$ where $(f_i^1)''(K)$ exists, $(f_i^1)''(K)<0$ holds. 
If $N(K)$ is constant on the complete interior of~$I_1$ (that is, on $(2\sqrt{a_i\gamma_i}+b_i,\sqrt{a_i\gamma_i}+b_i+C_i)$ or $(2\sqrt{a_i\gamma_i}+b_i, K_i^{\max})$, depending on the two possible cases for $I_1$), the desired monotonicity behaviour of $f_i^1$ over $I_1$ follows. 
Otherwise, let $\alpha_1<\alpha_2<\cdots < \alpha_k$ denote the different values of $b_j+p_j, j \in N^+$ which lie in the interior of $I_1$. Define $\alpha_0:=2\sqrt{a_i\gamma_i}+b_i$ and $\alpha_{k+1}:=\sup{I_1}$ (that is, $\alpha_{k+1}=\sqrt{a_i\gamma_i}+b_i+C_i$ or $\alpha_{k+1}=K_i^{\max}$). 
Then, $N(K)$ is constant on the intervals $(\alpha_{\ell-1},\alpha_{\ell}]$ for $\ell\in \{1,\ldots,k+1\}$. 
For each $\ell\in \{1,\ldots,k\}$, the set $N(K)$ increases immediately after $\alpha_{\ell}$, that is, $N(\alpha_{\ell})\subsetneq N(\alpha_{\ell}+\varepsilon)$ holds for any $\varepsilon>0$. In particular, $N(\alpha_{\ell}+\varepsilon)=N(\alpha_{\ell})\cup\{j \in N^+: b_j+p_j=\alpha_{\ell}\}$ holds for all $0<\varepsilon\leq\alpha_{\ell+1}-\alpha_{\ell}$. 
We now show that for any $\ell\in \{1,\ldots,k\}$, the slope of~$f_i^1$ decreases at $\alpha_{\ell}$, whereby we mean that
\[
(f_i^1)_{+}'(\alpha_{\ell})<(f_i^1)_{-}'(\alpha_{\ell})
\]
holds, with $(f_i^1)_{+}'(\alpha_{\ell})$ and $(f_i^1)_{-}'(\alpha_{\ell})$ denoting the right and left derivative of $f_i^1$ at $\alpha_{\ell}$, respectively. 
This implies the desired monotonicity behaviour of $f_i^1$ over $I_1$. 
Analyzing the left and right derivative of $f_i^1$ at $\alpha_{\ell}$ yields
\[
(f_i^1)_{-}'(\alpha_{\ell})=1-\sum_{j \in N(\alpha_{\ell})}{\frac{(2\alpha_{\ell}-b_j-p_j-b_i-2\sqrt{a_i\gamma_i})z_j}{a_j}}
\]
and 
\begin{align*}
(f_i^1)_{+}'(\alpha_{\ell})&=1-\sum_{j \in N(\alpha_{\ell})\cup\{j\in N^+: b_j+p_j=\alpha_{\ell}\}}{\frac{(2\alpha_{\ell}-b_j-p_j-b_i-2\sqrt{a_i\gamma_i})z_j}{a_j}}\\
&=(f_i^1)_{-}'(\alpha_{\ell}) -\sum_{j\in N^+: b_j+p_j=\alpha_{\ell}}{\frac{(2\alpha_{\ell}-\alpha_{\ell}-b_i-2\sqrt{a_i\gamma_i})z_j}{a_j}} \\
&=(f_i^1)_{-}'(\alpha_{\ell})-(\alpha_{\ell}-b_i-2\sqrt{a_i\gamma_i})\cdot\sum_{j \in N^+: b_j+p_j=\alpha_{\ell}}{\frac{z_j}{a_j}}.
\end{align*}
Since $\alpha_{\ell}$ lies in the interior of $I_1$, we get $\alpha_{\ell}> 2\sqrt{a_i\gamma_i}+b_i$ and therefore the desired inequality $(f_i^1)_{+}'(\alpha_{\ell})<(f_i^1)_{-}'(\alpha_{\ell})$, completing the proof for the case $\min\{b_j+p_j: j \in N^+\}<2\sqrt{a_i\gamma_i}+b_i$. 
The remaining case is $2\sqrt{a_i\gamma_i}+b_i \leq \min\{b_j+p_j: j \in N^+\} <\sqrt{a_i\gamma_i}+b_i+C_i$, which implies that $N(K)=\emptyset$ holds in $I_1$ for $K \leq \min\{b_j+p_j: j \in N^+\}$, and $N(K)\neq \emptyset$ holds in $I_1$ for $K >\min\{b_j+p_j: j \in N^+\}$. 
We can now obviously combine the arguments of the other two cases to obtain the desired monotonicity behaviour of $f_i^1$ also for this case. 
This completes the proof of statement~~\ref{lemma:BR3_1} of the lemma. 

%%%%%%%%%%%%%%%%%%%%%%%%%%%%%%%%%%%%%%%%%%%%%%%%%%%%%%%%%%%%%%%%%%%%%%%%%%%%%%%%%%%%%%%%%%%%%%%%%%%%%%
Now we turn to statement~\ref{lemma:BR3_2}, thus we assume that \hyperlink{P_2}{$(\text{P}_i^2)$} is feasible. Let $I_2$ denote the feasible set of \hyperlink{P_2}{$(\text{P}_i^2)$}. 
Then, either $I_2=[\frac{a_i\gamma_i}{C_i}+b_i+C_i, K_i^{\max})$ or $I_2=(\sqrt{a_i\gamma_i}+b_i+C_i,K_i^{\max})$ holds, depending on whether $C_i< \sqrt{a_i\gamma_i}$ holds or not. 
But in both cases, $I_2$ is an interval with \textit{positive} length, so that there exists $K\in I_2$ with $K>\frac{a_i\gamma_i}{C_i}+b_i+C_i$, which implies $f_i^2(K)>0$. Therefore, if \hyperlink{P_2}{$(\text{P}_i^2)$} has an optimal solution~$K^*_2$, we get $f_i^2(K^*_2)>0$. 
Furthermore note that the objective function
\[
f_i^2(K)=\overline x_i(K)\cdot(C_i-\frac{a_i\gamma_i}{K-b_i-C_i})=\left(1-\sum_{j \in N(K)}{\frac{(K-b_j-p_j)z_j}{a_j}}\right)\cdot(C_i-\frac{a_i\gamma_i}{K-b_i-C_i})
\]
of \hyperlink{P_2}{$(\text{P}_i^2)$} is continuous over $(b_i+C_i,\R)$. 
Using this, we can show that \hyperlink{P_2}{$(\text{P}_i^2)$} has \textit{at least one} optimal solution, if $C_i\leq \sqrt{a_i\gamma_i}$ holds: Due to the theorem of Weierstrass, $f_i^2$ attains its maximum over $[\frac{a_i\gamma_i}{C_i}+b_i+C_i,K_i^{\max}]$, the closure of $I_2$. But since $f_i^2(\frac{a_i\gamma_i}{C_i}+b_i+C_i)=0=f_i^2(K_i^{\max})$, and $f_i^2(K)>0$ for any $K \in (\frac{a_i\gamma_i}{C_i}+b_i+C_i, K_i^{\max})$, the maximum is not attained at $K=\frac{a_i\gamma_i}{C_i}+b_i+C_i$ or $K=K_i^{\max}$, which shows that $f_i^2$ also attains its maximum over $I_2$. 
Thus, if $C_i\leq \sqrt{a_i\gamma_i}$ holds, \hyperlink{P_2}{$(\text{P}_i^2)$} has at least one optimal solution. 
To complete the proof of statement~\ref{lemma:BR3_2}, it remains to show that \hyperlink{P_2}{$(\text{P}_i^2)$} has \textit{at most one} optimal solution (in the general case). 
As in the proof of statement~\ref{lemma:BR3_1} of the lemma, we proof this by showing that $f_i^2$ exhibits a certain monotonicity behaviour over $I_2$, namely: Either $f_i^2$ is strictly decreasing over $I_2$, or strictly increasing up to a unique point, and strictly decreasing afterwards. Note that $f_i^2$ cannot be strictly increasing over $I_2$, due to the continuity of $f_i^2$ and the fact that $f_i^2(K_i^{\max})=0<f_i^2(K)$ holds for any $K$ in the interior of $I_2$. 
The remaining proof is very similar to the proof of statement~\ref{lemma:BR3_1}. 
First, $f_i^2$ is twice differentiable on any open interval where $N(K)$ is constant. The first and second derivative of $f_i^2$ then are
\begin{align*}
(f_i^2)'(K) &=-\sum_{j \in N(K)}{\frac{z_j}{a_j}}\cdot(C_i-\frac{a_i\gamma_i}{K-b_i-C_i}) + \overline x_i(K)\cdot \frac{a_i\gamma_i}{(K-b_i-C_i)^2} \quad \text{and}\\
(f_i^2)''(K)&=- \sum_{j \in N(K)}{\frac{z_j}{a_j}} \cdot \frac{a_i\gamma_i}{(K-b_i-C_i)^2}-\sum_{j \in N(K)}{\frac{z_j}{a_j}} \cdot \frac{a_i\gamma_i}{(K-b_i-C_i)^2} \\
&\phantom{=}\ +\overline x_i(K)\cdot \frac{-2a_i\gamma_i}{(K-b_i-C_i)^3}\\
&= -\frac{2a_i\gamma_i}{(K-b_i-C_i)^2}\cdot \left(\sum_{j \in N(K)}{\frac{z_j}{a_j}}+\frac{\overline x_i(K)}{K-b_i-C_i}\right).
\end{align*}
Since $\overline x_i(K)>0$ for all $K\in I_2$, we conclude that for all $K\in I_2$ where $(f_i^2)''(K)$ exists, $(f_i^2)''(K)<0$ holds. 
If $N(K)$ is constant on the complete interior of~$I_2$, the desired monotonicity behaviour of $f_i^2$ over $I_2$ follows, otherwise let $\beta_1<\beta_2<\cdots < \beta_k$ denote the different values of $b_j+p_j, j \in N^+$ which lie in the interior of $I_2$. 
 We show that the slope of $f_i^2$ decreases at $\beta_{\ell}$, i.e. 
%\[
$
(f_i^2)_+'(\beta_{\ell})< (f_i^2)_-'(\beta_{\ell})
$ 
%\]
holds, which implies the desired monotonicity behaviour of $f_i^2$ over $I_2$. 
Analyzing the left and right derivative yields 
\begin{align*}
(f_i^2)_-'(\beta_{\ell}) 
&=-\sum_{j \in N(\beta_{\ell})}{\frac{z_j}{a_j}}\cdot(C_i-\frac{a_i\gamma_i}{\beta_{\ell}-b_i-C_i}) + \overline x_i(\beta_{\ell})\cdot \frac{a_i\gamma_i}{(\beta_{\ell}-b_i-C_i)^2} \quad \text{ and}
\end{align*}
\begin{align*}
(f_i^2)_+'(\beta_{\ell}) &=-\sum_{j \in N(\beta_{\ell})\cup \{j\in N^+: b_j+p_j=\beta_{\ell}\}}{\frac{z_j}{a_j}}\cdot(C_i-\frac{a_i\gamma_i}{\beta_{\ell}-b_i-C_i})+ \overline x_i(\beta_{\ell})\cdot \frac{a_i\gamma_i}{(\beta_{\ell}-b_i-C_i)^2} \\
&=(f_i^2)_-'(\beta_{\ell}) -\sum_{j\in N^+: b_j+p_j=\beta_{\ell}}{\frac{z_j}{a_j}}\cdot(C_i-\frac{a_i\gamma_i}{\beta_{\ell}-b_i-C_i})
\end{align*}
Since $\beta_{\ell}$ lies in the interior of $I_2$, we get $C_i-\frac{a_i\gamma_i}{\beta_{\ell}-b_i-C_i} > 0$, and thus the desired inequality $(f_i^2)_+'(\beta_{\ell})< (f_i^2)_-'(\beta_{\ell})$, 
completing the proof of statement~\ref{lemma:BR3_2} of the lemma.
%%%%%%%%%%%%%%%%%%%%%%%%%%%%%%%%%%%%%%%%%%%%%%%%%%%%%%%%%%%%%%%%%%%%%%%%%%%%%%%%%%%%%%%%%%%%%%

Finally we show statement~\ref{lemma:BR3_3} of the lemma, so assume that $K^*_1$ and $K^*_2$ are the optimal solutions of \hyperlink{P_1}{$(\text{P}_i^1)$} and \hyperlink{P_2}{$(\text{P}_i^2)$}. 
Since \hyperlink{P_1}{$(\text{P}_i^1)$} and \hyperlink{P_2}{$(\text{P}_i^2)$} have to be feasible, $\sqrt{a_i\gamma_i}\leq C_i$ and $\sqrt{a_i\gamma_i}+b_i+C_i<K_i^{\max}$ holds. 
This implies that the feasible set of \hyperlink{P_1}{$(\text{P}_i^1)$} is $I_1=[2\sqrt{a_i\gamma_i}+b_i,\sqrt{a_i\gamma_i}+b_i+C_i]$ and the feasible set of \hyperlink{P_2}{$(\text{P}_i^2)$} is $I_2=(\sqrt{a_i\gamma_i}+b_i+C_i,K_i^{\max})$. 
Let $\bar K:=\sqrt{a_i\gamma_i}+b_i+C_i$. Then, 
\begin{align*}
f_i^1(\bar K)&=\overline x_i(\sqrt{a_i\gamma_i}+b_i+C_i)\cdot(\sqrt{a_i\gamma_i}+b_i+C_i-b_i-2\sqrt{a_i\gamma_i}) \\
&= \overline x_i(\sqrt{a_i\gamma_i}+b_i+C_i)\cdot(C_i-\sqrt{a_i\gamma_i}) \\
&= \overline x_i(\sqrt{a_i\gamma_i}+b_i+C_i)\cdot\left(C_i-\frac{a_i\gamma_i}{\sqrt{a_i\gamma_i}+b_i+C_i-b_i-C_i} \right) \\
&=f_i^2(\bar K)
\end{align*}
holds. 
If additionally the slope of $f_i^1$ in $\bar K$ is greater than or equal to the slope of $f_i^2$ in $\bar K$, whereby we mean that
$
(f_i^1)_-'(\bar K)\geq (f_i^2)_+'(\bar K)
$
holds, we get $f_i^1(K^*_1)=f_i^1(\bar K)<f_i^2(K^*_2)$ from our analysis of $f_i^1$ and $f_i^2$ in the proofs of the statements~\ref{lemma:BR3_1} and~\ref{lemma:BR3_2} (note that $f_i^2$ is strictly increasing on the interval $(\sqrt{a_i\gamma_i}+b_i+C_i,K_2^*]$).
The remaining inequality for the slopes follows from
\begin{align*}
(f_i^1)_-'(\bar K) 
&= 1-\sum_{j \in N(\bar K)}{\frac{(2\bar K-b_j-p_j-b_i-2\sqrt{a_i\gamma_i})z_j}{a_j}} \\
&= 1-\sum_{j \in N(\bar K)}{\frac{(2C_i+b_i-b_j-p_j)z_j}{a_j}} \quad \text{and} \\
(f_i^2)_+'(\bar K) &= -\sum_{j \in N(\bar K)\cup \{j\in N^+: b_j+p_j=\bar K\}}{\frac{z_j}{a_j}}\cdot(C_i-\frac{a_i\gamma_i}{\bar K-b_i-C_i})+ \overline x_i(\bar K)\cdot \frac{a_i\gamma_i}{(\bar K-b_i-C_i)^2} \\
&= -\sum_{j \in N(\bar K)\cup \{j\in N^+: b_j+p_j=\bar K\}}{\frac{z_j}{a_j}}\cdot(C_i-\sqrt{a_i\gamma_i})+ \overline x_i(\bar K)\cdot 1 \\
&= -\sum_{j \in N(\bar K)}{\frac{z_j}{a_j}}\cdot(C_i-\sqrt{a_i\gamma_i}) -\sum_{j\in N^+: b_j+p_j=\bar K}{\frac{z_j}{a_j}}\cdot(C_i-\sqrt{a_i\gamma_i}) \\
& \phantom{=} \ +1-\sum_{j \in N(\bar K)}{\frac{(\bar K-b_j-p_j)z_j}{a_j}} \\
&=1-\sum_{j \in N(\bar K)}{\frac{(\bar K-b_j-p_j+C_i-\sqrt{a_i\gamma_i})z_j}{a_j}}-\sum_{j\in N^+: b_j+p_j=\bar K}{\frac{z_j}{a_j}}\cdot(C_i-\sqrt{a_i\gamma_i}) \\
&=1-\sum_{j \in N(\bar K)}{\frac{(2C_i+b_i-b_j-p_j)z_j}{a_j}}-\sum_{j\in N^+: b_j+p_j=\bar K}{\frac{z_j}{a_j}}\cdot(C_i-\sqrt{a_i\gamma_i}) \\
&= (f_i^1)_-'(\bar K) -\sum_{j\in N^+: b_j+p_j=\bar K}{\frac{z_j}{a_j}}\cdot(C_i-\sqrt{a_i\gamma_i}) \\
&\leq (f_i^1)_-'(\bar K),
\end{align*}
where the inequality follows from $\sqrt{a_i\gamma_i}\leq C_i$. 
\end{proof}

\subsection{The Characterization}\label{subsec:3}
The following theorem provides a complete characterization of the best response correspondence. 
We will make use of this characterization several times during the rest of the paper.
\begin{theorem}\label{theo:BR1}
For a firm $i\in N$ and fixed strategies $s_{-i}=(z_{-i},p_{-i}) \in S_{-i}$ of the other firms, the set $\BR_i=\BR_i(s_{-i})$ of best responses of firm $i$ to $s_{-i}$ is given as indicated in Table~\ref{table_characterization}, where the first column contains $\BR_i$ and the second column contains the conditions on $s_{-i}$ under which $\BR_i$ has the stated form. For $j=1,2$, $K^*_j$ denotes the unique optimal solution of problem \hyperlink{P_1}{$(\text{P}_i^j)$}, if this problem has an optimal solution.

Furthermore, if $\BR_i(s_{-i})$ consists of a \emph{unique} best response $s_i=(z_i,p_i)$ of firm~$i$ to~$s_{-i}$, we get $z_i>0$ and $\Pi_i\klammer{s_i,s_{-i}}>0$.
\begin{table}[h!]
\begin{center}
\begin{tabular}{l c@{\hspace{0.25cm}} l}
\toprule
$\{(z_i,p_i)\}= \BR_i$ & & conditions \\ \midrule
$\emptyset$ & & $z_{-i}=0$  \\ \addlinespace[0.75em]
$\{(0,p_i): 0\leq p_i\leq C_i\}$ & & $z_{-i}\neq 0$, \hyperlink{P_1}{$(\text{P}_i^1)$} and \hyperlink{P_2}{$(\text{P}_i^2)$} infeasible  \\ \addlinespace[0.75em]
$\left\{\left(\frac{a_i \cdot \overline x_i(K^*_2)}{K^*_2-b_i-C_i}, C_i\right)\right\}$ & & \parbox{8.5cm}{$z_{-i}\neq 0$, \hyperlink{P_2}{$(\text{P}_i^2)$} has an optimal solution} \\\addlinespace[0.75em]
$\left\{\left(\sqrt{\frac{a_i}{\gamma_i}}\cdot \overline x_i(K^*_1), K^*_1-\sqrt{a_i\gamma_i}-b_i\right)\right\}$ & &  \parbox{8.5cm}{$z_{-i}\neq 0$, \hyperlink{P_1}{$(\text{P}_i^1)$} feasible, \hyperlink{P_2}{$(\text{P}_i^2)$} has no optimal solution}\\\addlinespace[0.25em]
\bottomrule
\end{tabular}
\caption{Characterization of $\BR_i$.}
\label{table_characterization}
\end{center}
\end{table}

\end{theorem}
\begin{proof}
Note that if $\BR_i(s_{-i})$ consists of a \emph{unique} best response $s_i=(z_i,p_i)$ of firm~$i$ to~$s_{-i}$, then $z_i>0$ and $\Pi_i\klammer{s_i,s_{-i}}>0$ hold: Otherwise, any strategy $(z_i',p_i') \in \{0\}\times[0,C_i]$ is a best response, too, contradicting the uniqueness assumption. 

Now turn to the proof of the characterization. 
We show that the case distinction covers all possible cases, and that the given representation for $\BR_i$ is correct for each case.
If $z_{-i}=0$, Lemma~\ref{lemma:BR0} shows $\BR_i=\emptyset$. 
For the rest of the proof, assume $z_{-i}\neq 0$. 
Then, firm~$i$ has at least one best response to $s_{-i}$, since $\BR_i$ can be described as the set of optimal solution of the problem \hyperlink{P_i}{$(\text{P}_i)$}, and this problem has an optimal solution (as shown in the beginning of Subsection~\ref{subsec:2}). 
If \hyperlink{P_1}{$(\text{P}_i^1)$} and \hyperlink{P_2}{$(\text{P}_i^2)$} are both infeasible, Lemma~\ref{lemma:BR2} implies that each best response $(z_i,p_i)$ fulfills $z_i=0$. Therefore, $\BR_i=\{0\}\times [0,C_i]$. 
For the remaining proof, assume that at least one of \hyperlink{P_1}{$(\text{P}_i^1)$} and \hyperlink{P_2}{$(\text{P}_i^2)$} is feasible. 

First consider the case that \hyperlink{P_2}{$(\text{P}_i^2)$} has an optimal solution. 
It follows from~\ref{lemma:BR3_2} of Lemma~\ref{lemma:BR3} that the solution is unique, and, if $K^*_2$ denotes this unique solution, that $f_i^2(K^*_2)>0$. 
Now let $(z_i,p_i)$ be an arbitrary best response of player~$i$ to $s_{-i}$. We need to show that $(z_i,p_i)=(\frac{a_i\overline x_i(K_2^*)}{K_2^*-b_i-C_i}, C_i)$ holds. 
First, statement~\ref{lemma:BR1_2} of Lemma~\ref{lemma:BR1} shows that $\Pi_i(z_i,p_i)\geq f_i^2(K^*_2)>0$. Thus $z_i>0$ holds, since $z_i=0$ yields a profit of $0$. 
Note that either \hyperlink{P_1}{$(\text{P}_i^1)$} is infeasible, or it has a unique optimal solution $K^*_1$ with $f_i^1(K^*_1)<f_i^2(K^*_2)\leq \Pi_i(z_i,p_i)$ (see~\ref{lemma:BR3_1} and~\ref{lemma:BR3_3} of Lemma~\ref{lemma:BR3}). In both cases, Lemma~\ref{lemma:BR2} yields $(z_i,p_i)=(\frac{a_i\overline x_i(K_2^*)}{K_2^*-b_i-C_i}, C_i)$, as desired. 

Now assume that \hyperlink{P_2}{$(\text{P}_i^2)$} does not have an optimal solution (either \hyperlink{P_2}{$(\text{P}_i^2)$} is infeasible, or it is feasible, but the maximum is not attained). 
We first show that \hyperlink{P_1}{$(\text{P}_i^1)$} is feasible. If \hyperlink{P_2}{$(\text{P}_i^2)$} is infeasible, \hyperlink{P_1}{$(\text{P}_i^1)$} is feasible since we assumed that at least one of the two problems is feasible. Otherwise \hyperlink{P_2}{$(\text{P}_i^2)$} is feasible, but does not have an optimal solution. Then, $C_i>\sqrt{a_i\gamma_i}$ follows from~\ref{lemma:BR3_2} of Lemma~\ref{lemma:BR3}, and $\sqrt{a_i\gamma_i}+b_i+C_i<K_i^{\max}$ follows since \hyperlink{P_2}{$(\text{P}_i^2)$} is feasible. Together, $2\sqrt{a_i\gamma_i}+b_i<\sqrt{a_i\gamma_i}+b_i+C_i<K_i^{\max}$ holds, showing that \hyperlink{P_1}{$(\text{P}_i^1)$} is feasible. 
By~\ref{lemma:BR3_1} of Lemma~\ref{lemma:BR3} we then get that \hyperlink{P_1}{$(\text{P}_i^1)$} has a unique optimal solution $K^*_1$. Furthermore, each best response $(z_i,p_i)$ with $z_i>0$ fulfills $(z_i,p_i)=(\sqrt{a_i/\gamma_i}\cdot\overline x_i(K_1^*), K_1^*-\sqrt{a_i\gamma_i}-b_i)$ (see Lemma~\ref{lemma:BR2}). To complete the proof, we need to show that there is no best response $(z_i,p_i)$ with $z_i=0$. 
This follows from Lemma~\ref{lemma:BR1} if $f_i^1(K^*_1)>0$. Thus it remains to show that $f_i^1(K^*_1)>0$ holds. Assume, by contradiction, that $f_i^1(K^*_1)=0$. This implies that $K^*_1=2\sqrt{a_i\gamma_i}+b_i$ is the only feasible solution for~\hyperlink{P_1}{$(\text{P}_i^1)$}, which in turn yields $\sqrt{a_i\gamma_i}=C_i$ and $\sqrt{a_i\gamma_i}+b_i+C_i<K_i^{\max}$. But this implies that \hyperlink{P_2}{$(\text{P}_i^2)$} is feasible and has an optimal solution (by~\ref{lemma:BR3_2} of Lemma~\ref{lemma:BR3}), contradicting our assumption that \hyperlink{P_2}{$(\text{P}_i^2)$} does not have an optimal solution. 
\end{proof}

\subsection{Discussion}\label{subsec:4}
We now briefly discuss consequences of the characterization of the best reponse
correspondences with respect to applying Kakutani's fixed point theorem (see~\cite{Kakutani41}), or related results such as~\cite{Fan52} and~\cite{Glicksberg52}.  
Kakutani's theorem in particular requires that for each firm $i$ and each vector $s_{-i}=(z_{-i},p_{-i})$ of strategies of the other firms, the set $\BR_i(s_{-i})$ of best responses is nonempty  and convex. But as we have seen in Lemma~\ref{lemma:BR0}, the set $\BR_i(s_{-i})$ can be empty, namely if $z_{-i}=0$. On the other hand, a strategy profile with $z_{-i}=0$ for some firm $i$ will of course never be a PNE. 

A first natural approach to overcome the problem of empty best responses is the following. 
Given a strategy profile $s=(z,p)$ such that $z_{-i}=0$ for some firm $i$, 
redefine, for each such firm~$i$, the set $\BR_i(s_{-i})$ by some \textit{suitable} nonempty  convex set.  ``Suitable'' here means that the correspondence $\BR_i$ 
has a closed graph, and at the same time, $s$ must not be a fixed point of the global best response correspondence 
$\BR$ (where $\BR(s):=\{s'\in S: s_i'\in \BR_i(s_{-i}) \text{ for each } i\in N\}$). 
But unfortunately, these two goals are not compatible: 
For the strategy profile $s=(z,p)$ with $(z_i,p_i)=(0,C_i)$ for all firms $i$, 
the closed graph property requires $(0,C_i)\in \BR_i(s_{-i})$ for all $i$, which implies that $s$ is a fixed point of $\BR$. 

Another intuitive idea is to consider a game in which each firm has an initial capacity of some $\eps>0$. If this game has a PNE for each $\eps$, the limit for $\eps$ going to zero should be a PNE for our original capacity and price competition game. 
For the game with at least $\eps$ capacity, one can also characterize the best response correspondences (now by optimal solutions of \textit{three} optimization problems), but the main problem is that an analogue of Lemma~\ref{lemma:BR3} may not hold anymore.
As a consequence, it is not clear if the best responses are always convex, 
and we again do not know how to apply Kakutani's theorem. 
Instead, we show existence of PNE by using a result of McLennan et al.~\cite{McLennan11}, see Section~\ref{section:existence}.\footnote{Perhaps interesting, McLennan et al.  identify a non-trivial restriction of the players' non-equilibrium strategies so that they can eventually apply Kakutani's theorem.}

\section{Existence of Equilibria}\label{section:existence}
In this section, we show that each capacity and price competition game has a PNE.  
A frequently used tool to show existence of PNE is Kakutani's fixed point theorem. 
But, as discussed in Subsection~\ref{subsec:4}, we cannot directly apply this result to show existence of PNE.
Furthermore, the existence theorem of Reny~\cite{Reny99} can also not be used, since a capacity and price competition game is not quasiconcave in general (see Example~\ref{example_quasiconcave}). 
Instead, we turn to another existence result due to McLennan, Monteiro and Tourky~\cite{McLennan11}. 
They introduced a concept called \textit{C-security} and they showed that if the game is $C$-secure at each strategy profile which is not a PNE, then a PNE exists. 
Informally, the game is $C$-secure at a strategy profile $s$ if there is a vector $\alpha\in \R^n$ satisfying the following two properties: First, each firm $i$ has some securing strategy for $\alpha_i$ which is robust to small deviations of the other firms, i.e. firm $i$ always achieves a profit of at least $\alpha_i$ by playing this strategy even if the other firms slightly deviate from their strategies in $s_{-i}$. 
The second property requires for each slightly perturbed strategy profile $s'$ resulting from $s$, that there is at least one firm $i$ such that her perturbed strategy $s_i'$ can (in some sense) be strictly separated from all strategies achieving a profit of $\alpha_i$, so in particular from all her securing strategies. 
One can think of firm $i$ being ``not happy'' with her perturbed strategy $s_i'$ since she could achieve a higher profit. 
This already indicates the connection between a strategy profile which is not a PNE, and $C$-security. 
We will see that for certain strategy profiles, $\alpha_i$ can be chosen as the profit that firm $i$ gets by playing a best response to $s_{-i}$.\footnote{More precisely, we need to choose $\alpha_i$ slightly smaller than the profit of a best response.} Then, firm $i$'s securing strategies for $\alpha_i$ are related to her set of best responses, and we need to ``strictly separate'' these best responses from $s_i$. At this point, our characterization of best responses in Theorem~\ref{theo:BR1} becomes useful.

We now formally describe McLennan et al.'s result in our context. 
First of all, note that they consider games with compact convex strategy sets and bounded profit functions. In a capacity and price competition game, the strategy set $S_i=\{(z_i,p_i): 0\leq z_i, 0\leq p_i\leq C_i\}$ of firm $i$ is not compact a priori. But since $z_i$ will never be larger than $C_i/\gamma_i$ in any best response, and thus in any PNE (see the discussion on page~\pageref{bound_capacities}), we can redefine $S_i:=\{(z_i,p_i): 0 \leq z_i \leq C_i/\gamma_i, 0 \leq p_i \leq C_i\}$ without changing the set of PNE of the game, for any firm $i$. Furthermore, this also does not change the best responses, so Theorem~\ref{theo:BR1} continues to hold.
Using the redefined strategies, for any firm $i$ and any strategy profile $s$, the profit of firm $i$ is bounded by $-C_i\leq \Pi_i(s)\leq C_i$. 
For a strategy profile $s\in S$, firm $i \in N$ and $\alpha_i \in \R$ let
\[
B_i(s,\alpha_i):=\{s_i'\in S_i: \Pi_i\klammer{s_i',s_{-i}}\geq \alpha_i\} \ \text{ and } \ C_i(s,\alpha_i):=\conv B_i(s,\alpha_i),
\]
where $\conv B_i(s,\alpha_i)$ denotes the convex hull of $B_i(s,\alpha_i)$.
\begin{definition}
A firm $i$ can \textit{secure a profit $\alpha_i \in \R$ on $S' \subseteq S$}, if there is some $s_i \in S_i$ such that $s_i \in B_i(s',\alpha_i)$ for all $s' \in S'$. We say that firm $i$ can \textit{secure $\alpha_i$ at $s \in S$}, if she can secure $\alpha_i$ on $U\cap S$ for some open set $U$ with $s \in U$. 
\end{definition} 

\begin{definition}\label{def:Csecure}
The game is \textit{C-secure on $S' \subseteq S$}, if there is an $\alpha \in \R^n$ such that the following conditions hold:
\begin{enumerate}
	\item[(i)]\hypertarget{i}{} Every firm $i$ can secure $\alpha_i$ on $S'$.
	\item[(ii)]\hypertarget{ii}{} For any $s' \in S'$, there exists some firm $i$ with $s_i' \notin C_i(s',\alpha_i)$. 
\end{enumerate}
The game is \textit{C-secure at $s \in S$}, if it is $C$-secure on $U\cap S$ for some open set $U$ with $s \in U$.
\end{definition}
\noindent
We can now state the existence result of McLennan et al.:
\begin{theorem}[Proposition 2.7 in~\cite{McLennan11}]\label{theo:mclennan}
If the game is C-secure at each $s\in S$ that is not a PNE, then the game has a PNE.
\end{theorem}
We now turn to capacity and price competition games and show the existence of a PNE by using Theorem~\ref{theo:mclennan}, i.e., we show that if a given strategy profile $s=(z,p)$ is not a PNE, then the game is $C$-secure at $s$. 
To this end, we distinguish between the two cases that there are at least two firms $i$ with $z_i>0$ (Lemma~\ref{lemma:cont}), or not (Lemma~\ref{lemma:discont}). 
Both lemmata together then imply the desired existence result. 
Note that the mentioned case distinction is equivalent to the case distinction that each firm~$i$ has a best response for $s_{-i}$, or there is at least one firm $i$ with $\BR_i(s_{-i})=\emptyset$ (see Theorem~\ref{theo:BR1}). 
 
We start with the case that all best responses exist.
The proof of the following lemma follows an argument in~\cite[p. 1647f]{McLennan11} where McLennan et al. show that Theorem~\ref{theo:mclennan} implies the existence result of Nishimura and Friedman~\cite{Nishimura81}.
\begin{lemma}\label{lemma:cont}
Let $s=(z,p)\in S$ be a strategy profile which is not a PNE. 
Assume that there are at least two firms $i\in N$ such that $z_i>0$ holds.
Then the game is C-secure at~$s$. 
\end{lemma}
\begin{proof}
We first introduce some notation used in this proof.
Let $S'\subseteq S$ be a subset of the strategy profiles and $i\in N$.
By $S'_i\subseteq S_i$, we denote the projection of $S'$ into $S_i$, the set of firm~$i$'s strategies, and $S'_{-i}\subseteq S_{-i}$ denotes the projection of $S'$ into $S_{-i}=\times_{j\in N\setminus \{i\}}{S_j}$, the set of strategies of the other firms. 
Note that since $z$ has at least two positive entries $z_i$, all strategy profiles $s'=(z',p')$ in a sufficiently small open neighbourhood of $s$ also have at least two entries $z_i'>0$. In the following, whenever we speak of an open set $U$ containing $s$, we implicitly require $U$ small enough to fulfill this property. 
Furthermore, since it is clear that we are only interested in the elements of $U$ which are strategy profiles, we simply write $U$ instead of $U\cap S$. Consequently, $s'\in U$ denotes a strategy profile contained in $U$. 
Now we turn to the actual proof.

Since we assumed that at least two firms have positive capacity at $s=(z,p)$, we get that $z_{-i}\neq 0$ holds for each firm $i$. Thus, by Theorem~\ref{theo:BR1}, each firm $i$ has a best response to $s_{-i}$, and the set of best responses either is a singleton, or consists of all strategies $(0,p_i)$ for $0\leq p_i\leq C_i$.
Since $s$ is not a PNE, there is at least one firm $j$ such that $s_j=(z_j,p_j)$ is not a best response, i.e. either $s_j\neq s_j^*$ for the unique best response $s_j^*$, or $z_j>0$, and all best responses $s_j^*=(z_j^*,p_j^*)$ fulfill $z_j^*=0$. 
In both cases it is clear that there is a hyperplane $H$ which strictly separates $s_j$ from the set of best responses to $s_{-j}$. 

We now turn to the properties in Definition~\ref{def:Csecure}. 
For each firm $i$, let $s_i^*$ be a best response of firm~$i$ to $s_{-i}$ and $\beta_i:=\Pi_i\klammer{s_i^*,s_{-i}}$. We know from Theorem~\ref{theo:continuity2} that $\Pi_i$ is continuous at $(s_i^*,s_{-i})$ for each firm~$i$. 
Therefore, for each $\eps>0$, there is an open set $U(\eps)$ containing $s$ such that $\Pi_i\klammer{s_i^*,s_{-i}'}\geq \beta_i-\eps$ for each $s'\in U(\eps)$ and each firm $i$. 
That is, each firm $i$ can secure $\beta_i-\eps$ on $U(\eps)$. 
Now turn to the second property of Definition~\ref{def:Csecure} and consider firm $j$. We now show that there is an $\eps>0$ and an open set $U\subseteq U(\eps)$ containing $s$, such that for each $s'\in U$, 
the hyperplane $H$ (which strictly separates $s_j$ and $B_j(s,\beta_j)$)  
also strictly separates $s_j'$ and $B_j(s',\beta_j-\eps)$, thus $s_j' \notin C_j(s',\beta_j-\eps)$ (see Figure~\ref{pic:trennung} for an illustration). 
Since each firm $i$ can secure $\beta_i-\eps$ on $U\subseteq U(\eps)$, both properties of Definition~\ref{def:Csecure} are fulfilled, completing the proof. 

\begin{figure}[h!]
\centering
\begin{tikzpicture}[dot/.style={name=#1},
  extended line/.style={shorten >=-#1,shorten <=-#1},
  extended line/.default=1cm]
\node [dot=A] at (-3,2) {};
\node [dot=B] at (3,1) {};
%\node [dot=P,above] at (0,2.5) {$y$};

%\node[right] (t) at (0,2) {$\lambda^\intercal=x_0-y$};
\node[above] (r) at (A) {$H$};

\draw [thick] (A) -- (B);
%\draw [thick] ($(A)!(P)!(B)$) -- (P);

%\fill [red] ($(A)!(P)!(B)$) circle [radius=2pt];

    \filldraw[ultra thick, fill=green!80!black, opacity=0.4] (0,0) ellipse (2cm and 1cm);
    
     \filldraw[ultra thick, fill=green!80!black, opacity=0.4] (1,3) ellipse (2cm and 1cm);
  %   \fill (3,3) circle[radius=2pt];
 %   \fill (0,2.5) circle[radius=2pt] ;
  %  \node (t) at (0,0.75) {$x_0$};
    
     \node (t1) at (-0.75,0.5) {$s_j^*$};
      \fill (-0.5,0.5) circle[radius=1.5pt];
      \node (t2) at (1.5,2.75) {$s_j$};
       \fill (1.5,3) circle[radius=1.5pt];
       \node (sj') at (0.25,2.75) {$s'_j$};
       \fill (0.5,2.7) circle[radius=1.5pt];
       \node (t1) at (-2.25,0) {$V$};
      \node (t2) at (-1.25,3.5) {$U_j$};
      
       \draw[scale=0.65,fill=gray!30,rotate=-45] (0,0) to [out=140,in=90] (-1,-0.7)
    to [out=-90,in=240] (0.8,-1.15)
    to [out=60,in=-60] (2.1,1.2)
    to [out=120,in=90] (0.3,0.4)
    to [out=-90,in=20] (0.1,0.0)
    to [out=200,in=-40] (0,0);
   % \draw (-0.5,-0.5) -- (0.7,0.7);
  %  \fill (-0.5,-0.5) circle[radius=1.5pt];
  %  \fill (0.7,0.5) circle[radius=1.5pt];
    \node (B) at (0.5,-0.3) {$B_j(s',\beta_j-\varepsilon)$};
\end{tikzpicture}
\caption{Illustration of the proof construction for the case $B_j(s,\beta_j)=\{s_j^*\}$.
Note that it is not necessary that $s_j^*\in B_j(s',\beta_j-\varepsilon)$.
\vspace{-1em}}\label{pic:trennung}
\end{figure}
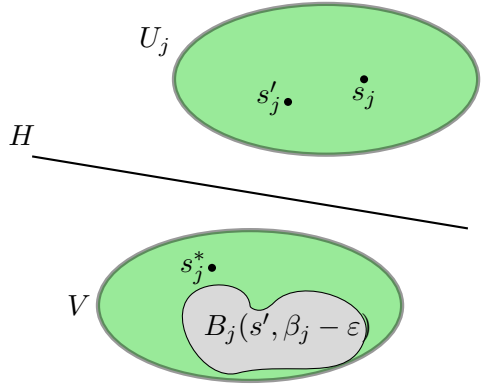
To this end, choose an open set $V$ containing $B_j(s,\beta_j)$ such that $H$ strictly separates $s_j$ and~$V$.
Since $S_j\setminus V$ is a compact set and $\Pi_j$ is continuous at $(\widetilde{s}_j,s_{-j})$ for all $\widetilde{s}_j\in S_j\setminus V$ (by Theorem~\ref{theo:continuity2}), we get that
$f(s_{-j}):=\max\{\Pi_j\klammer{\widetilde{s}_j,s_{-j}}: \widetilde{s}_j \in S_j \setminus V\}$ exists. Furthermore, since $B_j(s,\beta_j)\subset V$, we get $f(s_{-j})<\beta_j$.
Let $0<\eps<\beta_j-f(s_{-j})$, thus $f(s_{-j})<\beta_j-\eps$. 
Note that if we consider, for an open neighbourhood $\bar U$ of $s$ and for fixed $s_{-j}'\in \bar U_{-j}$, the problem of maximizing $\Pi_j\klammer{\widetilde{s}_j,s_{-j}'}$ subject to $\widetilde{s}_j \in S_j\setminus V$, Berge's theorem of the maximum~(\cite{berge63})  yields that $f(s_{-j}'):=\max\{\Pi_j\klammer{\widetilde{s}_j,s_{-j}'}: \widetilde{s}_j \in S_j \setminus V\}$ is a continuous function. 
Using the continuity of $f$, there is an open set $U\subseteq U(\eps)$ containing $s$ such that $f(s_{-j}')<\beta_j-\eps$ for all $s_{-j}' \in U_{-j}$. Additionally, let $U$ be small enough such  that $H$ strictly separates $U_j$ and $V$.
Now we have the desired properties: For each $s' \in U$ and for each $\widetilde{s}_j\in S_j\setminus V$, we get $\Pi_j\klammer{\widetilde s_j,s_{-j}'}<\beta_j-\eps$, thus $B_j(s',\beta_j-\eps)\subset V$. Since $s_j'\in U_j$ and $H$ strictly separates $U_j$ and $V$, we get that $H$ strictly separates $s_j'$ and $B_j(s',\beta_j-\eps)$, as desired. 
\end{proof}
It remains to analyze the strategy profiles $s=(z,p)$ with at most one positive $z_i$. Note that these profiles cannot be PNE.
\begin{lemma}\label{lemma:discont}
Let $s=(z,p)\in S$ be a strategy profile such that $z_i>0$ holds for at most one firm~$i$. Then the game is $C$-secure at $s$. 
\end{lemma}
\begin{proof}
We distinguish between the two cases that there is a firm with positive capacity, or all capacities are zero. 
In the former case, assume that $z_i>0$ for firm $i$, and $z_j=0$ for all $j\neq i$ hold. 
Choose $\alpha_i \in (C_i-\gamma_i z_i, C_i)$ and $0<\eps<\min\{z_i,\frac{\alpha_i+\gamma_iz_i-C_i}{\gamma_i}\}$. Then, there is an open set $U$ containing~$s$ such that firm $i$ can secure $\alpha_i$ on $U\cap S=:S'$ (note that by choosing $U$ sufficiently small, firm $i$ can secure each profit $<C_i$ on $S'$) \emph{and} $|z_i'-z_i|<\eps$ holds for each $s'=(z',p') \in S'$. 
For $j\neq i$, set $\alpha_j:=0$.  It is clear that each firm $j\neq i$ can secure $\alpha_j=0$ on $S'$ (by any strategy with zero capacity). In this way, property~\hyperlink{i}{(i)} of $C$-security is fulfilled. 
For property~\hyperlink{i}{(ii)}, let $s'=(z',p') \in S'$. We show that $s_i'\notin  C_i(s',\alpha_i)$ holds. To this end, note that any strategy $s_i^*=(z_i^*,p_i^*)\in B_i(s',\alpha_i)$, i.e. with $\Pi_i\klammer{s_i^*,s_{-i}'}\geq \alpha_i$, fulfills $z_i^*\leq z_i-\eps$, since, for $z_i^*> z_i-\eps>0$, we get
\[
\Pi_i\klammer{s_i^*,s_{-i}'}=x_i\klammer{s_i^*,s_{-i}'}p_i^*-\gamma_iz_i^* \leq C_i-\gamma_iz_i^*<C_i-\gamma_i(z_i-\eps)<\alpha_i,
\]
where the last inequality follows from the choice of $\eps$. 
Clearly, any strategy in $C_i(s',\alpha_i)$, i.e. any convex combination of strategies in $B_i(s',\alpha_i)$, then also has this property. Since $z_i'>z_i-\eps$, we get $(z_i',p_i')\notin C_i(s',\alpha_i)$, as desired. Thus we showed that the game is $C$-secure at $s$ if one firm has positive capacity.

Now turn to the case that $z_i=0$ for all $i\in N$. 
We distinguish between two further subcases, namely that there is a firm $i$ with $p_i<C_i$, or all prices are at their upper bounds. 
In the former case, let $i\in N$ with $p_i<C_i$, and 
choose $\alpha_i$ with $p_i<\alpha_i<C_i$. 
There is an open set $U$ containing~$s$ such that firm $i$ can secure $\alpha_i$ on $U\cap S=:S'$  \emph{and} $p_i'<\alpha_i$ holds for each $s'=(z',p')\in S'$. By setting $\alpha_j:=0$ for all $j\neq i$, property~\hyperlink{i}{(i)} of $C$-security is fulfilled. 
For property~\hyperlink{i}{(ii)}, let $s'=(z',p') \in S'$. We show that $s_i'\notin  C_i(s',\alpha_i)$ holds. 
Our assumptions about $S'$ yield $p_i'<\alpha_i$. 
On the other hand, any strategy $s_i^*=(z_i^*,p_i^*)$ with $\Pi_i\klammer{s_i^*,s_{-i}'}\geq \alpha_i$ obviously fulfills $p_i^*> \alpha_i$. In particular, this holds for any strategy in $C_i(s',\alpha_i)$, thus showing that $s_i' \notin C_i(s',\alpha_i)$. 
We conclude that the game is $C$-secure at $s$ for the case that all $z_i$ are zero \emph{and}  there is a firm $i$ with $p_i<C_i$. 

It remains to consider the case that $(z_i,p_i)=(0,C_i)$ holds for all firms $i$. 
For each firm $i$, choose $\alpha_i$ with $(1-\frac{a_i}{2(a_i+C_i)})C_i<\alpha_i<C_i$. Note that this implies
\begin{equation}\label{eq:existence1}
\frac{1}{2}<\frac{1}{2}+\frac{C_i-\alpha_i}{a_i}<\frac{\alpha_i}{C_i}.
\end{equation}
There is an open set $U$ containing~$s$ such that each firm $i$ can secure $\alpha_i$ on $U\cap S=:S'$ \emph{and} $z_i'<1$ holds for each profile $s'=(z',p')\in S'$. Thus, property~\hyperlink{i}{(i)} of $C$-security is fulfilled. For property~\hyperlink{ii}{(ii)}, let $s'=(z',p')\in S'$. 
In the following, $s_i^*=(z_i^*,p_i^*)$ denotes a strategy of firm~$i$ with $\Pi_i\klammer{s_i^*,s_{-i}'}\geq \alpha_i>0$. We say that $s_i^*$ \textit{achieves a profit of at least $\alpha_i$}. 
Obviously, $z_i^*>0$ and $p_i^*>\alpha_i$ hold. Furthermore, $x_i\klammer{s_i^*,s_{-i}'}> \frac{1}{2}$, since $\Pi_i\klammer{s_i^*,s_{-i}'}=x_i\klammer{s_i^*,s_{-i}'}p_i^*-\gamma_i z_i^*\geq \alpha_i$ implies
\begin{equation}\label{eq:existence1.5}
x_i\klammer{s_i^*, s'_{-i}} \geq \frac{\alpha_i+\gamma_iz_i^*}{p_i^*}> \frac{\alpha_i}{C_i}>\frac{1}{2},
\end{equation}
where the last inequality is due to \eqref{eq:existence1}. 
If $z_i'=0$ holds for a firm~$i$, then $s_i'\notin C_i(s',\alpha_i)$ holds, since any strategy $(z_i^*,p_i^*)$ achieving a profit of at least $\alpha_i>0$ fulfills $z_i^*>0$. Thus we can assume in the following that $z_i'> 0$ holds for all firms~$i$. Then, since $n\geq 2$, there is at least one firm $i$ with $x_i(s')\leq \frac{1}{2}$. 
We now show that $s_i'\notin C_i(s',\alpha_i)$ holds.  
If $z_i'=0$ or $p_i'\leq \alpha_i$ holds, $s_i'\notin C_i(s',\alpha_i)$ follows, since $z_i^*>0$ and $p_i^*>\alpha_i$ hold for any strategy $(z_i^*,p_i^*)$ achieving a profit of at least $\alpha_i$. 
Thus we can assume in the following that $z_i'>0$ and $p_i'>\alpha_i$ hold. 
If $x_i(s')=0$, the Wardrop equilibrium conditions yield $K(s')\leq b_i+p_i'$. Then, any strategy $\bar{s_i}=(\bar{z_i},\bar{p_i})$ with $\bar{p_i}\geq p_i'$ yields $x\klammer{\bar{s_i},s_{-i}'}=x(s')$, and thus $x_i\klammer{\bar{s_i},s_{-i}'}=0$ and $\Pi_i\klammer{\bar{s_i},s_{-i}'}\leq 0<\alpha_i$ hold. Therefore, $p_i^*<p_i'$ holds for any strategy $(z_i^*,p_i^*)$ achieving a profit of at least $\alpha_i$, and $s_i'\notin C_i(s',\alpha_i)$ follows. 
We can thus assume in the following that $x_i(s')>0$ holds. 
Summarizing, we can assume that the following inequalities are fulfilled: 
\begin{equation}
0<x_i\klammer{s_i',s'_{-i}}\leq \frac{1}{2}, \quad \alpha_i<p_i'\leq C_i \quad \text{and} \quad 0<z_i'<1.
\label{eq:existence2}
\end{equation}
We now show that each strategy $s_i^*=(z_i^*,p_i^*)$ with $\Pi_i\klammer{s_i^*, s_{-i}'}\geq \alpha_i$ fulfills $z_i^*>z_i'$, showing that $s_i'\notin C_i(s',\alpha_i)$ and completing the proof. 

Assume, by contradiction, that there is a strategy $s_i^*=(z_i^*,p_i^*)$ which achieves a profit of at least $\alpha_i$ and fulfills $z_i^*\leq z_i'$.  
For any strategy $\tilde{s_i}\in S_i$, write $x(\tilde{s_i}):=x\klammer{\tilde{s_i},s_{-i}'}$ and $K(\tilde{s_i}):=K\klammer{\tilde{s_i}, s_{-i}'}$. 
Now consider the strategy $\tilde{s_i}:=(z_i',\alpha_i)$. 
Assume, for the moment, that 
\begin{equation}
K(\tilde{s_i})\leq K(s_i^*)<K(s_i') 
\label{eq:existence2.5}
\end{equation}
holds (we prove~\eqref{eq:existence2.5} below). 
Using $K(\tilde{s_i})<K(s_i')$ then implies
$
\frac{a_i}{z_i'}x_i(\tilde{s_i})+b_i+\alpha_i<\frac{a_i}{z_i'}x_i(s_i')+b_i+p_i'. 
$
Reformulating this inequality and using \eqref{eq:existence2} and \eqref{eq:existence1} then yields
\begin{equation}\label{eq:existence3}
x_i(\tilde{s_i})<\frac{z_i'(p_i'-\alpha_i)}{a_i}+x_i(s_i')< \frac{C_i-\alpha_i}{a_i}+\frac{1}{2}<\frac{\alpha_i}{C_i}.
\end{equation}
The inequality $K(\tilde{s_i})\leq K(s_i^*)$ from~\eqref{eq:existence2.5} implies $x_j(\tilde{s_i})\leq x_j(s_i^*)$ for all firms $j \neq i$. Therefore, $x_i(\tilde{s_i})\geq x_i(s_i^*)$ holds. Using $x_i(s_i^*)>\frac{\alpha_i}{C_i}$ from~\eqref{eq:existence1.5} now leads to $x_i(\tilde{s_i})>\frac{\alpha_i}{C_i}$, which contradicts \eqref{eq:existence3}. 
To complete the proof, it remains to show \eqref{eq:existence2.5}. 
The property $K(s_i^*)<K(s_i')$ holds since $K(s_i^*)\geq K(s_i')$ would imply $x_j(s_i^*)\geq x_j(s_i')$ for all firms $j \neq i$, and thus $x_i(s_i^*)\leq x_i(s_i')$, but we know from~\eqref{eq:existence1.5} and~\eqref{eq:existence2} that $x_i(s_i^*)>1/2\geq x_i(s_i')$.  
To prove the other inequality in \eqref{eq:existence2.5}, assume, by contradiction, that $K(\tilde{s_i})> K(s_i^*)$. This implies $x_j(\tilde{s_i})\geq x_j(s_i^*)$ for all firms $j \neq i$, and thus $x_i(\tilde{s_i})\leq x_i(s_i^*)$. Together with $z_i^*\leq z_i'$ and $p_i^*>\alpha_i$, this leads to the following contradiction, and finally completes the proof:
\[
K(\tilde{s_i})\leq \frac{a_i}{z_i'}x_i(\tilde{s_i})+b_i+\alpha_i<\frac{a_i}{z_i^*}x_i(s_i^*)+b_i+p_i^*=K(s_i^*)<K(\tilde{s_i}).
\]
\end{proof}
Using Theorem~\ref{theo:mclennan} together with the Lemmata~\ref{lemma:cont} and~\ref{lemma:discont} now yields the existence of a PNE:
\begin{theorem}\label{theo:existence}
Every capacity and price competition game has a pure Nash equilibrium.
\end{theorem}
Note here that in any PNE $(z,p)$, there are at least two firms $i$ with $z_i>0$. 

We conclude this section with an example showing that in general, a capacity and price competition game is not quasiconcave (that is, it is not the case that for all $i\in N$ and all $s_{-i}\in S_{-i}$, the profit $\Pi_i(\cdot,s_{-i})$ is quasiconcave on $S_i$). Thus, the result of Reny~\cite{Reny99} cannot be used to show existence of PNE. 
\begin{example}\label{example_quasiconcave}
Consider a capacity and price competition game with two firms, $N=\{1,2\}$, and $a_1=1$, $a_2=2$, $b_1=b_2=1$, $C_1=C_2=10$, $\gamma_1=\gamma_2=\frac{1}{4}$. 
Assume that firm~1 chooses the strategy $s_1=(z_1,p_1)=(1,1)$, thus $\ell_1(x_1,z_1)+p_1=x_1+2$. 

We show that $\Pi_2(\cdot, s_1)$ is not quasiconcave on $S_2$. 
To this end, note that any strategy~$(z_2,p_2)$ with $z_2=0$, $p_2\in [0,10]$ yields a profit of 0, thus in particular the strategy~$(0,10)$. 
Furthermore, if firm~2 chooses $(z_2,p_2)=(2,1)$, this also results in a profit of $1/2-2/4=0$ (since $\ell_2(x_2,z_2)+p_2=x_2+2$, and thus the induced Wardrop flow is $x_1=x_2=1/2$). 
But the convex combination $\frac{1}{2}(0,10)+\frac{1}{2}(2,1)=(1,11/2)$ yields profit $-1/4<0$ (since $x_1=1, x_2=0$ is the induced Wardrop flow).  
By definition, this shows that $\Pi_2(\cdot, s_1)$ is not quasiconcave on $S_2$. 
\end{example}

\section{Uniqueness of Equilibria}\label{sec:unique}
As we have seen in the last section, a  capacity and price competition game always has a PNE. 
In this section we show that this equilbrium is \textit{essentially} unique. 
With \textit{essentially} we mean that if $(z,p)$ and $(z',p')$ are two different PNE, and $i\in N$ is a firm such that $(z_i,p_i)\neq (z_i',p_i')$, then $z_i=z_i'=0$ holds (and thus $p_i\neq p_i'$). 

For a PNE $s=(z,p)$, denote by $N^+(z,p):=\{i\in N: z_i>0\}$ the set of firms with positive capacity (note that $|N^+(z,p)|\geq 2$ and $N^+(z,p)=\{i\in N: x_i(s)>0\}$). 
For $i\in N^+(z,p)$, let \hyperlink{P_1}{$(\text{P}_i^1)(s_{-i})$} and \hyperlink{P_2}{$(\text{P}_i^2)(s_{-i})$} be the two auxiliary problems from Section~\ref{section:BR}.\footnote{In Section~\ref{section:BR}, we considered \textit{fixed} strategies $s_{-i}$, thus we just used \hyperlink{P_1}{$(\text{P}_i^1)$} and \hyperlink{P_2}{$(\text{P}_i^2)$} for the problems corresponding to $s_{-i}$. In this section, we need to consider different strategy profiles, thus we now write \hyperlink{P_1}{$(\text{P}_i^1)(s_{-i})$} and \hyperlink{P_2}{$(\text{P}_i^2)(s_{-i})$}, as well as $K_i^{\max}(s_{-i})$.}
By Lemma~\ref{lemma:BR2}, the routing cost $K(z,p)$ is an optimal solution of either \hyperlink{P_1}{$(\text{P}_i^1)(s_{-i})$} or \hyperlink{P_2}{$(\text{P}_i^2)(s_{-i})$}. We denote by $N^+_1(z,p)$ the set of firms $i\in N^+(z,p)$ such that $K(z,p)$ is an optimal solution of \hyperlink{P_1}{$(\text{P}_i^1)(s_{-i})$}, and $N^+_2(z,p)$ contains the firms $i\in N^+(z,p)$ such that $K(z,p)$ is an optimal solution of \hyperlink{P_2}{$(\text{P}_i^2)(s_{-i})$}. Thus $N^+(z,p)=N^+_1(z,p)\overset{.}{\cup}N^+_2(z,p)$. Throughout this section, we use the simplified notation $N'\setminus i$ instead of $N' \setminus \{i\}$ for any subset $N'\subseteq N$ of firms and  $i\in N'$.

Note that the proofs in this section are similar to the proofs that Johari et al.~\cite{JohariWR10} use to derive their uniqueness results. However, since our model includes price caps, some new ideas are required, in particular the decomposition of $N^+(z,p)$ in $N^+_1(z,p)\overset{.}{\cup}N^+_2(z,p)$. 

We first derive further necessary equilibrium conditions (by using the KKT conditions) which will become useful in the following analysis. 
\begin{lemma}\label{lemma:KKT}
Let $s=(z,p)$ be a PNE with $x:=x(z,p)$ and $K:=K(z,p)$. Let $i \in N^+:=N^+(z,p)$. 
If $p_i<C_i$ holds, then $z_i=\sqrt{\frac{a_i}{\gamma_i}}x_i$, $p_i=\frac{x_i}{\sum_{j\in N^+\setminus i}{\frac{z_j}{a_j}}}+\sqrt{a_i\gamma_i}$, and if $p_i=C_i$, then $z_i=\frac{a_ix_i}{K-b_i-C_i}$ and $\frac{C_i}{1+\frac{a_i}{z_i}\cdot \sum_{j\in N^+\setminus i}{\frac{z_j}{a_j}}}=\frac{\gamma_iz_i^2}{a_ix_i\sum_{j\in N^+\setminus i}{\frac{z_j}{a_j}}}$.
\end{lemma}
\begin{proof}
Since $(z,p)$ is a PNE, $p_j>0$ and $x_j>0$ holds for all $j \in N^+$, and $x_k=0$ holds for $k \notin N^+$. Furthermore, $(z_i,p_i)$ is a best response of firm $i$ to $s_{-i}$ and $K=\frac{a_jx_j}{z_j}+b_j+p_j$ holds for all $j \in N^+$.
Altogether we get that $(z_i,p_i,(x_j)_{j \in N^+})$ is an optimal solution for the following optimization problem (with variables $(z_i',p_i',(x_j')_{j \in N^+})$):
\begin{align*}
\qquad\qquad \max & \quad   x_i'p_i'-\gamma_iz_i'\\
\qquad\qquad \text{ subject to } & \quad 0 \leq p_i' \leq C_i\\
& \quad 0 < z_i' \\
& \quad \sum_{j \in N^+}{x_j'}=1 \\
& \quad x_j' \geq 0 \ \forall j \in N^+ \\
& \quad \frac{a_ix_i'}{z_i'}+b_i+p_i'=\frac{a_jx_j'}{z_j}+b_j+p_j \ \forall j \in N^+\setminus i.
\end{align*} 
It is easy to show that the LICQ is fulfilled for $(z_i,p_i,(x_j)_{j \in N^+})=(z_i,p_i,x_i,(x_j)_{j \in N^+ \setminus i})$ (see the end of the proof), thus the KKT conditions are fulfilled. 
We get the following equations:
\begin{align}
\gamma_i-\frac{a_ix_i}{z_i^2}\sum_{j \in N^+\setminus i}{\lambda_j} &=0 \tag{KKT1}\label{eq:1} \\
-x_i+\mu+\sum_{j \in N^+\setminus i}{\lambda_j} &=0 \tag{KKT2}\label{eq:2}\\
-p_i+\lambda+\frac{a_i}{z_i}\sum_{j \in N^+\setminus i}{\lambda_j} &=0 \tag{KKT3}\label{eq:3}\\
\lambda-\lambda_j\frac{a_j}{z_j}&=0 \ \forall j \in N^+\setminus i. \tag{KKT4}\label{eq:4}
\end{align}
We now distinguish between the two cases $p_i<C_i$ and $p_i=C_i$. 

In the first case, $\mu=0$ holds, and \eqref{eq:2} yields $x_i=\sum_{j \in N^+\setminus i}{\lambda_j}$. Using this, \eqref{eq:1} yields $z_i=\sqrt{\frac{a_i}{\gamma_i}}x_i$. Plugging this in \eqref{eq:3} leads to $p_i=\lambda+\frac{a_ix_i}{z_i}=\lambda+\sqrt{a_i\gamma_i}$. Using \eqref{eq:4}, i.e. $\lambda_j=\lambda\frac{z_j}{a_j}$ for all $j \in N^+ \setminus i$, together with \eqref{eq:2} yields $x_i=\lambda \sum_{j \in N^+\setminus i}{\frac{z_j}{a_j}}$, or equivalently, $\lambda=\frac{x_i}{\sum_{j \in N^+\setminus i}{\frac{z_j}{a_j}}}$. This shows $p_i=\frac{x_i}{\sum_{j \in N^+\setminus i}{\frac{z_j}{a_j}}}+\sqrt{a_i\gamma_i}$, as required. 

The other case is $p_i=C_i$. The formula for $z_i$ follows from $K=\frac{a_ix_i}{z_i}+b_i+C_i$. Plugging $\lambda_j=\lambda\frac{z_j}{a_j}$ for all $j \in N^+ \setminus i$ in \eqref{eq:1} and \eqref{eq:3} yields 
\[
\lambda=\frac{\gamma_iz_i^2}{a_ix_i\sum_{j \in N^+\setminus i}{\frac{z_j}{a_j}}}
\]
and 
\[
\lambda=C_i-\frac{a_i}{z_i}\lambda \cdot \sum_{j \in N^+\setminus i}{\frac{z_j}{a_j}} \Leftrightarrow \lambda=\frac{C_i}{1+\frac{a_i}{z_i}\sum_{j \in N^+\setminus i}{\frac{z_j}{a_j}}},
\]
which shows the desired equality.

It remains to show that the LICQ is fulfilled for $(z_i,p_i,(x_j)_{j \in N^+})=(z_i,p_i,x_i,(x_j)_{j \in N^+ \setminus i})$. 
Besides the gradients of the equalities, we have to take into account the gradient of the inequality $p_i' \leq C_i$ (all other inequalities are not tight). 
Thus consider 
\[
\renewcommand{\arraystretch}{1.25}
\alpha_1\cdot \begin{pmatrix} 0 \\ 1 \\ 0 \\ 0 \\[0.5em] \twelvevdots \\ 0 \end{pmatrix} 
+ \alpha_2 \cdot \begin{pmatrix} 0 \\ 0 \\ 1 \\ 1 \\[0.5em] \twelvevdots \\ 1 \end{pmatrix} 
+ \sum_{j\in N^+\setminus i}{\alpha_j \cdot \begin{pmatrix}-\frac{a_ix_i}{z_i^2} \\ 1 \\ \frac{a_i}{z_i} \\ 0 \\ \vdots \\ 0 \\ -\frac{a_j}{z_j} \\ 0 \\ \vdots \\ 0 \end{pmatrix}}=0,
\]
where all vectors are in $\R^{2+|N^+|}$. For $j\in N^+ \setminus i$ and the corresponding vector in the sum, the entry corresponding to $j$ is $-a_j/z_j$, and the entries corresponding to $N^+ \setminus \{i,j\}$ are all zero. 
We show that all $\alpha_k$ are zero. 
Considering the first row yields $-\frac{a_ix_i}{z_i^2}\cdot \sum_{j \in N^+ \setminus i}{\alpha_j}=0$ which yields  $\sum_{j \in N^+ \setminus i}{\alpha_j}=0$ since $-\frac{a_ix_i}{z_i^2}<0$. 
Using this, we get from the second row $\alpha_1+\sum_{j \in N^+ \setminus i}{\alpha_j}=\alpha_1=0$. 
The third row yields $\alpha_2+\frac{a_i}{z_i}\cdot \sum_{j \in N^+ \setminus i}{\alpha_j}=\alpha_2=0$. 
Finally, any row corresponding to $j\in N^+ \setminus \{i\}$ reads $\alpha_2+\alpha_j \cdot (-\frac{a_j}{z_j})=\alpha_j \cdot (-\frac{a_j}{z_j})=0$ and this shows $\alpha_j=0$ since $\frac{a_j}{z_j}>0$.
\end{proof}
In the next lemma, we introduce two functions $\Gamma_i^1$ and $\Gamma_i^2$ for each firm $i$ and derive useful properties of these functions. 
\begin{lemma}\label{lemma:unique1}
For each $i \in N$, define
\[
\Gamma_i^1: (\sqrt{a_i\gamma_i}+b_i,\infty) \rightarrow \R, \ \Gamma_i^1(\kappa):=\frac{\sqrt{a_i\gamma_i}}{\kappa -\sqrt{a_i\gamma_i}-b_i}
\] and 
\[
\Gamma_i^2: (b_i+C_i,\infty) \rightarrow \R, \ \Gamma_i^2(\kappa):=\frac{\frac{a_i\gamma_i}{C_i}}{\kappa -b_i-C_i}.
\] 
Furthermore, let $s=(z,p)$ be a PNE with $x:=x(z,p)$, cost $K:=K(z,p)$, and $N^+:=N^+(z,p)$ with $N^+_1:=N^+_1(z,p)$, $N^+_2:=N^+_2(z,p)$. 
Then: 
\begin{enumerate}[1.]
	\item\label{prop1} $\Gamma_i^1$ and $\Gamma_i^2$ are strictly decreasing functions.
	\item\label{prop2} If $i \in N^+_1$, then $\Gamma_i^1(K)=1-\frac{\frac{z_i}{a_i}}{\sum_{j\in N^+}{\frac{z_j}{a_j}}}<1$.
	\item\label{prop3} If $i \in N^+_2$, then $\Gamma_i^2(K)=1-\frac{\frac{z_i}{a_i}}{\sum_{j\in N^+}{\frac{z_j}{a_j}}}<1$.
	\item\label{prop4} $\sum_{i \in N^+_1}{\Gamma_i^1(K)}+\sum_{i \in N^+_2}{\Gamma_i^2(K)}=|N^+|-1$.
	\item\label{prop5} If $i\in N_1^+$ and there is a (different) PNE $s'=(z',p')$ with $i\in N^+_2(z',p')$, then $K<K':=K(z',p')$ and $\Gamma_i^1(K)>\Gamma_i^2(K')$.
\end{enumerate}
\end{lemma}
\begin{proof}
Statement~\ref{prop1} is clear from the definitions of $\Gamma_i^1$ and $\Gamma_i^2$, so turn to statement~\ref{prop2} and let $i\in N^+_1$. 
Lemma~\ref{lemma:BR2} yields $z_i=\sqrt{\frac{a_i}{\gamma_i}}x_i$, $p_i=K-\sqrt{a_i\gamma_i}-b_i$. 
Using Lemma~\ref{lemma:KKT}, we get $K-\sqrt{a_i\gamma_i}-b_i=\frac{x_i}{\sum_{j\in N^+\setminus i}{\frac{z_j}{a_j}}}+\sqrt{a_i\gamma_i}$, which is equivalent to
\[
K=\frac{x_i}{\sum_{j\in N^+\setminus i}{\frac{z_j}{a_j}}}+2\sqrt{a_i\gamma_i}+b_i.
\]
(Note that $p_i=C_i$ is possible (namely if $K=\sqrt{a_i\gamma_i}+b_i+C_i$), but Lemma~\ref{lemma:KKT} yields the stated equality also for this case.)
Now define
\[
B:={\sum_{j\in N^+}{\frac{z_j}{a_j}}}.
\]
Using $z_i=\sqrt{\frac{a_i}{\gamma_i}}x_i$, we rewrite $K$ as follows:
\begin{align*}
K &= \frac{x_i}{B-{\frac{z_i}{a_i}}}+2\sqrt{a_i\gamma_i}+b_i \\
&= \frac{x_i}{B-{\frac{x_i}{\sqrt{a_i\gamma_i}}}}+2\sqrt{a_i\gamma_i}+b_i \\
&= \frac{x_i+(B-\frac{x_i}{\sqrt{a_i\gamma_i}})\sqrt{a_i\gamma_i}}{B-\frac{x_i}{\sqrt{a_i\gamma_i}}}+\sqrt{a_i\gamma_i}+b_i \\
&=\frac{B\sqrt{a_i\gamma_i}}{B-\frac{x_i}{\sqrt{a_i\gamma_i}}}+\sqrt{a_i\gamma_i}+b_i 
\end{align*}
Using this we get statement~\ref{prop2}:
\[
\Gamma_i^1(K)=\frac{B-\frac{x_i}{\sqrt{a_i\gamma_i}}}{B}=1-\frac{\frac{z_i}{a_i}}{B}<1.
\]
For statement~\ref{prop3}, let $i \in N^+_2$. 
Lemma~\ref{lemma:BR2} and Lemma~\ref{lemma:KKT} imply
\[
\frac{C_i}{1+\frac{a_i}{z_i}\cdot \sum_{j\in N^+\setminus i}{\frac{z_j}{a_j}}}=\frac{\gamma_iz_i^2}{a_ix_i\sum_{j\in N^+\setminus i}{\frac{z_j}{a_j}}}.
\]
Rearranging and using the definition of $B$ yields 
\[
\frac{a_ix_i}{z_i}=
\frac{\gamma_iz_i(1+\frac{a_i}{z_i}\cdot \sum_{j\in N^+\setminus i}{\frac{z_j}{a_j}})}{C_i\cdot \sum_{j\in N^+\setminus i}{\frac{z_j}{a_j}}}=
\frac{\gamma_iz_i(1+\frac{a_i}{z_i}\cdot (B-\frac{z_i}{a_i}))}{C_i\cdot (B-\frac{z_i}{a_i})}=
\frac{\gamma_ia_iB}{C_i\cdot (B-\frac{z_i}{a_i})}.
\]
Using $K=\frac{a_ix_i}{z_i}+b_i+C_i$ then yields
\[
K=\frac{\gamma_ia_iB}{C_i\cdot (B-\frac{z_i}{a_i})}+b_i+C_i
\]
and thus statement~\ref{prop3} follows:
\[
\Gamma_i^2(K)=\frac{B-\frac{z_i}{a_i}}{B}=1-\frac{\frac{z_i}{a_i}}{B}<1.
\]
Statement~\ref{prop4} now follows from the statements~\ref{prop2} and~\ref{prop3}:
\begin{align*}
\sum_{i\in N_1^+}{\Gamma_i^1(K)}+\sum_{i\in N^+_2}{\Gamma_i^2(K)}=&\sum_{i\in N_1^+}{\left(1-\frac{\frac{z_i}{a_i}}{B}\right)}+\sum_{i\in N^+_2}{\left(1-\frac{\frac{z_i}{a_i}}{B}\right)} \\
=&
|N^+|-\frac{1}{B}\cdot \sum_{j\in N^+}{\frac{z_j}{a_j}}=|N^+|-1.
\end{align*}
It remains to show statement~\ref{prop5}. Let $i\in N_1^+\cap N_2^+(z',p')$. Since $i\in N_1^+$, the cost~$K$ is in particular feasible for \hyperlink{P_1}{$(\text{P}_i^1)(s_{-i})$}, thus $2\sqrt{a_i\gamma_i}+b_i \leq K\leq \sqrt{a_i\gamma_i}+b_i+C_i$ holds. Analogously, using $i\in N_2^+(z',p')$, the cost $K'$ is feasible for \hyperlink{P_2}{$(\text{P}_i^2)(s_{-i}')$}, therefore $\sqrt{a_i\gamma_i}+b_i+C_i<K'<K_i^{\max}(s'_{-i})$. 
Together we get $K\leq \sqrt{a_i\gamma_i}+b_i+C_i< K'$. 
It remains to show $\Gamma_i^1(K)>\Gamma_i^2(K')$. 
By definition of $N_2^+(z',p')$, the cost~$K'$ is an optimal solution for problem \hyperlink{P_2}{$(\text{P}_i^2)(s_{-i}')$} and in particular (see~\ref{lemma:BR3_1} and~\ref{lemma:BR3_3}  of Lemma~\ref{lemma:BR3}) yields a better objective function value than $\sqrt{a_i\gamma_i}+b_i+C_i$ in \hyperlink{P_1}{$(\text{P}_i^1)(s_{-i}')$} (note that $\sqrt{a_i\gamma_i}+b_i+C_i$ is feasible for \hyperlink{P_1}{$(\text{P}_i^1)(s_{-i}')$} since $2\sqrt{a_i\gamma_i}+b_i \leq K\leq \sqrt{a_i\gamma_i}+b_i+C_i<K'<K_i^{\max}(s'_{-i})$). 
If $\overline x_i$ denotes the function occurring in the definitions of \hyperlink{P_1}{$(\text{P}_i^1)(s_{-i}')$} and \hyperlink{P_2}{$(\text{P}_i^2)(s_{-i}')$}, we thus get
\begin{align*}
\overline x_i(\sqrt{a_i\gamma_i}+b_i+C_i)\cdot(C_i-\sqrt{a_i\gamma_i})&<\overline x_i(K')\cdot(C_i-\frac{a_i\gamma_i}{K'-b_i-C_i}) \\
&\leq \overline x_i(\sqrt{a_i\gamma_i}+b_i+C_i)\cdot(C_i-\frac{a_i\gamma_i}{K'-b_i-C_i}),
\end{align*}
where the last inequality follows from $\sqrt{a_i\gamma_i}+b_i+C_i<K'$ and the fact that 
$\overline x_i$ is a decreasing function. 
Since $\overline x_i(\sqrt{a_i\gamma_i}+b_i+C_i)>0$, we get 
\[
C_i-\sqrt{a_i\gamma_i}<C_i-\frac{a_i\gamma_i}{K'-b_i-C_i},
\] 
or equivalently
\[
\sqrt{a_i\gamma_i} > \frac{a_i\gamma_i}{K'-b_i-C_i},
\]
which yields 
\[
\frac{\sqrt{a_i\gamma_i}}{C_i} > \frac{\frac{a_i\gamma_i}{C_i}}{K'-b_i-C_i}=\Gamma_i^2(K').
\]
Note that $2\sqrt{a_i\gamma_i}+b_i\leq K\leq \sqrt{a_i\gamma_i}+b_i+C_i$ and thus 
\[
\frac{1}{K-\sqrt{a_i\gamma_i}-b_i}\geq \frac{1}{C_i},
\]
which leads to
\[
\Gamma_i^1(K)=\frac{\sqrt{a_i\gamma_i}}{K-\sqrt{a_i\gamma_i}-b_i}\geq \frac{\sqrt{a_i\gamma_i}}{C_i}>\Gamma_i^2(K'),
\]
as desired.
\end{proof}
Now we turn to the desired uniqueness of the equilibrium. 
We start with the following lemma.
\begin{lemma}\label{lemma:unique2}
For a fixed subset $N^+$ of the firms and a fixed disjoint decomposition $N^+=N_1^+\overset{.}{\cup}N_2^+$, there is essentially at most one PNE $(z,p)$ such that $N^+(z,p)=N^+$, $N_1^+(z,p)=N_1^+$ and $N_2^+(z,p)=N_2^+$. 
\end{lemma}
\begin{proof}
Assume that there are two PNE $(z,p)$ and $(z',p')$ with the described properties, i.e. $N^+(z,p)=N^+(z',p')=N^+$, $N_1^+(z,p)=N_1^+(z',p')=N_1^+$ and $N_2^+(z,p)=N_2^+(z',p')=N_2^+$. Let $x:=x(z,p)$ and $x':=x(z',p')$ with costs $K:=K(z,p)$ and $K':=K(z',p')$. 
We show that $(z_i,p_i)=(z_i',p_i')$ holds for all $i\in N^+$, showing that $(z,p)$ and $(z',p')$  are essentially the same. 

First note that $K=K'$ holds, since $f(\kappa):=\sum_{i\in N_1^+}{\Gamma_i^1(\kappa)}+\sum_{i\in N_2^+}{\Gamma_i^2(\kappa)}$ is a strictly decreasing function in $\kappa$ and
\[
f(K)=\sum_{i \in N^+_1}{\Gamma_i^1(K)}+\sum_{i \in N^+_2}{\Gamma_i^2(K)}=|N^+|-1=\sum_{i \in N^+_1}{\Gamma_i^1(K')}+\sum_{i \in N^+_2}{\Gamma_i^2(K')}=f(K')
\]
holds from~\ref{prop4} in Lemma~\ref{lemma:unique1}. 
This implies $p_i=p_i'$ for all $i\in N^+$, since $p_i=K-\sqrt{a_i\gamma_i}-b_i=K'-\sqrt{a_i\gamma_i}-b_i=p_i'$ holds for $i\in N_1^+$, and $p_i=C_i=p_i'$ for $i\in N_2^+$. 

If $B:=\sum_{j \in N^+}{\frac{z_j}{a_j}}=\sum_{j \in N^+}{\frac{z_j'}{a_j}}=:B'$ holds, we
also get $z_i=z_i'$ for all $i\in N^+$, since~\ref{prop2} of Lemma~\ref{lemma:unique1} yields 
\[
z_i=(1-\Gamma_i^1(K))a_iB=(1-\Gamma_i^1(K'))a_iB'=z_i'
\]
for all $i\in N_1^+$ and~\ref{prop3} of Lemma~\ref{lemma:unique1} yields 
\[
z_i=(1-\Gamma_i^2(K))a_iB=(1-\Gamma_i^2(K'))a_iB'=z_i'
\]
for all $i\in N_2^+$. 

It remains to show $B=B'$. 
First consider $i\in N_1^+$. Using $z_i=\sqrt{a_i/\gamma_i}\cdot x_i$ and $z_i'=\sqrt{a_i/\gamma_i}\cdot x_i'$, as well as~\ref{prop2} of Lemma~\ref{lemma:unique1}, yields
\[
x_i/B=(1-\Gamma_i^1(K))\sqrt{a_i\gamma_i}=(1-\Gamma_i^1(K'))\sqrt{a_i\gamma_i}=x_i'/B'.
\] 
For $i\in N_2^+$, we use $\frac{z_i}{a_i}=\frac{x_i}{K-b_i-C_i}$ and $\frac{z_i'}{a_i}=\frac{x_i'}{K'-b_i-C_i}$ and~\ref{prop3} of Lemma~\ref{lemma:unique1} to achieve
\[
x_i/B=(1-\Gamma_i^2(K))(K-b_i-C_i)=(1-\Gamma_i^2(K'))(K'-b_i-C_i)=x_i'/B'.
\] 
Altogether we have $x_i/B=x_i'/B'$ for all $i\in N^+$. Using that $\sum_{i\in N^+}{x_i}=1=\sum_{i\in N^+}{x_i'}$ yields $B=B'$, as desired:
\begin{align*}
\frac{1}{B} = \sum_{i\in N^+}{\frac{x_i}{B}} = \sum_{i\in N^+}{\frac{x_i'}{B'}}=\frac{1}{B'}.
\end{align*}
\end{proof}
In the previous lemma, we showed that given a fixed subset $N^+\subseteq N$ and a fixed disjoint decomposition $N^+=N_1^+\overset{.}{\cup}N_2^+$, there is at most one PNE $(z,p)$ such that $N^+(z,p)=N^+$, $N_1^+(z,p)=N_1^+$ and $N_2^+(z,p)=N_2^+$. 
Next, we strengthen this result by showing that for a fixed subset $N^+\subseteq N$, there is at most one PNE $(z,p)$ with $N^+(z,p)=N^+$ (independently of the decomposition of $N^+$).
\begin{lemma}\label{lemma:unique3}
For a fixed subset $N^+$ of the firms, there is essentially at most one PNE $(z,p)$ with $N^+(z,p)=N^+$. 
\end{lemma}
\begin{proof}
Assume, by contradiction, that there are two essentially different PNE $(z,p)$ and $(\overline z,\overline p)$ with $N^+(z,p)=N^+=N^+(\overline z, \overline p)$. Let $N_1^+:=N_1^+(z,p)$, $N_2^+:=N_2^+(z,p)$ and $\overline N_1^+:=N_1^+(\overline z,\overline p)$, $\overline N_2^+:=N_2^+(\overline z,\overline p)$. Further denote $x:=x(z,p)$, $K:=K(z,p)$ and $\overline x:=x(\overline z,\overline p)$, $\overline K:=K(\overline z,\overline p)$. 

Lemma~\ref{lemma:unique2} yields that the decompositions of $N^+$ have to be different. Without loss of generality, there is a firm $j\in N_1^+ \setminus \overline N_1^+$. Since $j\in \overline N_2^+$, statement~\ref{prop5} of Lemma~\ref{lemma:unique1} yields $K<\overline K$. 
The existence of a firm $i \in \overline{N}_1^+ \setminus N_1^+$ leads (by the same argumentation) to the contradiction $\overline K<K$, thus $\overline{N}_1^+ \subsetneq N_1^+$ and $N_2^+\subsetneq \overline{N}_2^+$ hold and we can write (using~\ref{prop4} of Lemma~\ref{lemma:unique1})
\[
|N^+|-1=
\sum_{i \in N_1^+}{\Gamma_i^1(K)}+\sum_{i \in N_2^+}{\Gamma_i^2(K)}=\sum_{i \in N_1^+\setminus \overline{N}_1^+}{\Gamma_i^1(K)}+\sum_{i \in \overline{N}_1^+}{\Gamma_i^1(K)}+\sum_{i \in N_2^+}{\Gamma_i^2(K)}
\]
and 
\[
|N^+|-1=
\sum_{i \in \overline{N}_1^+}{\Gamma_i^1(\overline{K})}+\sum_{i \in \overline{N}_2^+}{\Gamma_i^2(\overline{K})}=
\sum_{i \in \overline{N}_1^+}{\Gamma_i^1(\overline{K})}+\sum_{i \in \overline{N}_2^+\setminus N_2^+}{\Gamma_i^2(\overline{K})}+\sum_{i \in {N}_2^+}{\Gamma_i^2(\overline{K})}.
\]
Using that $K<\overline{K}$ and $\sum_{i\in N_2^+}{\Gamma_i^2(\kappa)}$ is a decreasing function in $\kappa$ yields
\[
\sum_{i \in N_2^+}{\Gamma_i^2(K)} \geq \sum_{i \in N_2^+}{\Gamma_i^2(\overline{K})}.
\]
Furthermore, $\sum_{i \in \overline{N}_1^+}{\Gamma_i^1(\kappa)}$ is also decreasing in $\kappa$, thus
\[
\sum_{i \in \overline{N}_1^+}{\Gamma_i^1(K)} \geq \sum_{i \in \overline{N}_1^+}{\Gamma_i^1(\overline{K})}.
\]
Finally, statement~\ref{prop5} of Lemma~\ref{lemma:unique1} yields $\Gamma_i^1(K)>\Gamma_i^2(\overline K)$ for all $i\in N_1^+ \setminus \overline N_1^+=\overline N_2^+ \setminus N_2^+\neq \emptyset$, thus
\[
\sum_{i \in {N}_1^+\setminus  \overline{N}_1^+}{\Gamma_i^1(K)} > \sum_{i \in \overline{N}_2^+\setminus N_2^+}{\Gamma_i^2(\overline{K})}
\]
holds. 
Altogether we get the contradiction
\begin{align*}
|N^+|-1=&\sum_{i \in N_1^+\setminus \overline{N}_1^+}{\Gamma_i^1(K)}+\sum_{i \in \overline{N}_1^+}{\Gamma_i^1(K)}+\sum_{i \in N_2^+}{\Gamma_i^2(K)} \\
>& \sum_{i \in \overline{N}_1^+}{\Gamma_i^1(\overline{K})}+\sum_{i \in \overline{N}_2^+\setminus N_2^+}{\Gamma_i^2(\overline{K})}+\sum_{i \in {N}_2^+}{\Gamma_i^2(\overline{K})} \\
=& |N^+|-1,
\end{align*}
which completes the proof. 
\end{proof}

For the desired uniqueness of the PNE, it remains to show that there is at most one set $N^+$ such that a PNE $(z,p)$ with $N^+(z,p)=N^+$ exists. 
To this end, we first show that each firm~$i$ has a certain threshold $K^*_i$ such that, for any PNE $(z,p)$, firm $i$ has $z_i>0$ if and only if $K(z,p)>K^*_i$. 
\begin{lemma}\label{lemma:unique5}
For each $i\in N$, define
\[
K^*_i:=\begin{cases}
\frac{a_i\gamma_i}{C_i}+b_i+C_i, & \text{ if } \sqrt{a_i \gamma_i}>C_i, \\
2\sqrt{a_i\gamma_i}+b_i, & \text{else.}
 \end{cases}
\]
Then, for any PNE $s=(z,p)$ and any firm $i\in N$, it holds that $z_i>0$ if and only if $K(z,p)>K^*_i$. 
\end{lemma}
\begin{proof}
Let $s=(z,p)$ be a PNE with $x:=x(z,p)$, $K:=K(z,p)$, and $i\in N$.

First assume that $z_i=0$. Since $(z,p)$ is a PNE, the strategy $(z_i,p_i)=(0,p_i)$ is a best response of firm $i$ to $s_{-i}$. As we have seen in Theorem~\ref{theo:BR1}, this is equivalent to the fact that both problems \hyperlink{P_1}{$(\text{P}_i^1)(s_{-i})$} and \hyperlink{P_2}{$(\text{P}_i^2)(s_{-i})$} are infeasible. 
Note that $K=K_i^{\max}(s_{-i})$ holds due to the definition of $K_i^{\max}(s_{-i})$ (cf. page~\pageref{page_K_i_max}) and 
\[
0=1-\sum_{j\in N: x_j>0}{x_j}=1-\sum_{j \in N\setminus i: z_j>0, b_j+p_j<K}{\frac{(K-b_j-p_j)z_j}{a_j}},
\]
where we used that $\{j\in N: x_j>0\}=\{j \in N\setminus i: z_j>0, b_j+p_j<K\}$. 
To show $K\leq K^*_i$, we have to distinguish between the two cases $\sqrt{a_i \gamma_i}>C_i$ and $\sqrt{a_i \gamma_i}\leq C_i$.
First consider $\sqrt{a_i \gamma_i}>C_i$, thus $K^*_i=a_i\gamma_i/C_i+b_i+C_i$. 
Since \hyperlink{P_2}{$(\text{P}_i^2)(s_{-i})$} is infeasible and $\sqrt{a_i\gamma_i}+b_i+C_i<a_i\gamma_i/C_i+b_i+C_i$, we get the desired inequality $K=K_i^{\max}(s_{-i})\leq a_i\gamma_i/C_i+b_i+C_i=K^*_i$.
Now consider $\sqrt{a_i \gamma_i}\leq C_i$, i.e. $K^*_i=2\sqrt{a_i\gamma_i}+b_i$. Since \hyperlink{P_1}{$(\text{P}_i^1)(s_{-i})$} is infeasible and $2\sqrt{a_i\gamma_i}+b_i\leq \sqrt{a_i\gamma_i}+b_i+C_i$, we get $K=K_i^{\max}(s_{-i})\leq 2\sqrt{a_i\gamma_i}+b_i=K^*_i$, as desired. 
We have seen that $z_i=0$ implies $K\leq K^*_i$, or, equivalently, $K>K^*_i$ implies $z_i>0$.

It remains to show the other direction, i.e. $z_i>0$ implies $K> K^*_i$. We consider the two cases $\sqrt{a_i \gamma_i}>C_i$ and $\sqrt{a_i \gamma_i}\leq C_i$ and use our results from Lemma~\ref{lemma:BR2} and Theorem~\ref{theo:BR1}. 
If $\sqrt{a_i \gamma_i}>C_i$, thus $K^*_i=\frac{a_i\gamma_i}{C_i}+b_i+C_i$, the cost $K$ is an optimal solution for problem \hyperlink{P_2}{$(\text{P}_i^2)(s_{-i})$} with positive objective function value (note that \hyperlink{P_1}{$(\text{P}_i^1)(s_{-i})$} is infeasible), therefore $K^*_i=\frac{a_i\gamma_i}{C_i}+b_i+C_i<K$.
In the second case, i.e. $\sqrt{a_i \gamma_i}\leq C_i$ and $K^*_i=2\sqrt{a_i\gamma_i}+b_i$, the cost $K$ either is optimal for \hyperlink{P_1}{$(\text{P}_i^1)(s_{-i})$}, or optimal for \hyperlink{P_2}{$(\text{P}_i^2)(s_{-i})$}, and has positive objective function value in both cases. We get the desired property, since $K^*_i=2\sqrt{a_i\gamma_i}+b_i<K$ holds for the first case, and $K^*_i=2\sqrt{a_i\gamma_i}+b_i\leq \sqrt{a_i\gamma_i}+b_i+C_i<K$ holds for the second case, completing the proof.
\end{proof}

We can now show the remaining result for the desired uniqueness of PNE.
\begin{lemma}\label{lemma:unique4}
There is at most one subset $N^+$ of the firms such that a PNE $(z,p)$ with $N^+(z,p)=N^+$ exists. 
\end{lemma}
\begin{proof}
Assume, by contradiction, that there are two different subsets $N^+$ und $\overline{N}^+$ with corresponding PNE $(z,p)$ and $(\overline z,\overline p)$, such that $N^+(z,p)=N^+$ and $N^+(\overline z, \overline p)=\overline{N}^+$. 
Let $N_1^+\overset{.}{\cup} N_2^+$ and $\overline{N}_1^+\overset{.}{\cup}\overline{N}_2^+$ be the decompositions of $N^+$ and $\overline{N}^+$, that is, $N^+_1(z,p)=N^+_1$, $N^+_2(z,p)=N^+_2$, $N^+_1(\overline z,\overline p)=\overline N^+_1$ and $N^+_2(\overline z,\overline p)=\overline N^+_2$. 
Finally, denote $x:=x(z,p)$, $K:=K(z,p)$ and $\overline x:=x(\overline z,\overline p)$, $\overline K:=K(\overline z,\overline p)$.

Using Lemma~\ref{lemma:unique5}, we can assume w.l.o.g. that $K< \overline K$ and $N^+\subsetneq \overline{N}^+$. 
Then, $N_2^+\subseteq \overline{N}_2^+$ holds, since the existence of a firm $i \in N_2^+ \setminus \overline N_2^+$, i.e. $i \in N_2^+ \cap \overline N_1^+$, leads to the contradiction $\overline K<K$ by statement~\ref{prop5} of Lemma~\ref{lemma:unique1}. 
Furthermore, if there is a firm $i\in N_1^+ \setminus \overline N_1^+$, i.e. $i \in N_1^+ \cap \overline N_2^+$, statement~\ref{prop5} of Lemma~\ref{lemma:unique1} yields $\Gamma_i^1(K) >\Gamma_i^2(\overline K)$. 
Finally, $\Gamma_i^1(\overline K)<1$ holds for all $i \in \overline N_1^+$, and $\Gamma_i^2(\overline K)<1$ holds for all $i\in \overline N_2^+$ (see~\ref{prop2} and \ref{prop3} of Lemma~\ref{lemma:unique1}). 
Altogether, this leads to the following contradiction, and completes the proof (where we additionally use $K <\overline K$, and the statements~\ref{prop1} and~\ref{prop4} of Lemma~\ref{lemma:unique1}):
\begin{align*}
|\overline N^+|-1 =& \sum_{i \in \overline N_1^+}{\Gamma_i^1(\overline K)}+\sum_{i \in \overline N_2^+}{\Gamma_i^2(\overline K)} \\
=& \sum_{i \in \overline N_1^+ \cap N_1^+}{\Gamma_i^1(\overline K)}+\sum_{i \in \overline N_1^+ \setminus N^+}{\Gamma_i^1(\overline K)}+\sum_{i \in \overline N_2^+ \cap N_1^+}{\Gamma_i^2(\overline K)}+\sum_{i \in N_2^+}{\Gamma_i^2(\overline K)}+\sum_{i \in \overline N_2^+ \setminus N^+}{\Gamma_i^2(\overline K)}\\
<& \sum_{i \in N_1^+ \cap \overline N_1^+}{\Gamma_i^1(K)}+\sum_{i \in N_1^+\cap \overline N_2^+ }{\Gamma_i^1(K)}+\sum_{i \in N_2^+}{\Gamma_i^2(K)}+|\overline N^+|-|N^+|\\
=& \sum_{i \in N_1^+}{\Gamma_i^1(K)}+\sum_{i \in N_2^+}{\Gamma_i^2(K)}+|\overline N^+|-|N^+|
= |N^+|-1+|\overline N^+|-|N^+|=|\overline N^+|-1.
\end{align*}
\end{proof}
Together with the existence result in Theorem~\ref{theo:existence}, the preceding Lemmata~\ref{lemma:unique3} and~\ref{lemma:unique4} show that there is essentially one PNE. 
\begin{theorem}\label{theo:uniqueness}
Every capacity and price competition game has an essentially unique pure Nash equilibrium, i.e. if $(z,p)$ and $(z',p')$ are two different PNE and $i\in N$ is a firm such that $(z_i,p_i)\neq (z_i',p_i')$, then $z_i=z_i'=0$ holds.
\end{theorem}

\section{Quality of Equilibria}
In the last section, we showed that a capacity and price competition game has an (essentially) unique PNE.
Now we show that this PNE can be arbitrarily inefficient compared to a \textit{social optimum}.

Define the \textit{social cost} $C(z,p)$ of a strategy profile $s=(z,p)$ as 
\[
C(z,p)=\begin{cases}
\sum_{i\in\{1,\ldots,n\}: z_i>0}{\left(\ell_i(x_i(s),z_i)x_i(s)+\gamma_iz_i\right)}, & \text{ if $\sum_{i=1}^{n}{z_i}>0$,} \\
\infty, & \text{ else.}
 \end{cases}
\]
The function $C(z,p)$  measures utilitarian social welfare
over firms and customers  (the price component cancels out).
Common notions to measure the quality of equilibria are the Price of Anarchy (PoA) and the Price of Stability (PoS), which are defined as the worst case ratios of the cost of a worst, respectively best, pure Nash equilibrium, and a social optimum. 
Since a capacity and price competition game $G$ has an essentially unique PNE, all PNE of $G$ have the same social cost. If we denote this cost by $C(\PNE(G))$, we get 
\[
\PoA=\PoS=\sup_G \frac{C(\PNE(G))}{\OPT(G)},
\]
where $\OPT(G)$ denotes the minimum social cost in $G$ (compared to all possible strategy profiles). 

The following theorem shows that PoA and PoS are unbounded for capacity and price competition games:
\begin{theorem}\label{theo:quality}
$\PoA=\PoS=\infty$. The bound is attained even for games with only two firms.
\end{theorem}
\begin{proof}
Consider the capacity and price competition game $G_M$ with 
$
n=2 \text{ and } a_1=\gamma_1=C_1=~1, b_1=0 \text{ and } a_2=\gamma_2=C_2=M, b_2=0, \text{ where } M\geq 1.
$ 
By $z_1=1,p_1=z_2=p_2=0$, we get a profile with social cost 2, thus OPT$(G_M)\leq 2$. 
We will now show that $C(\PNE(G_M))>M$ holds, which implies
\[
\PoA=\PoS\geq \frac{C(\PNE(G_M))}{\OPT(G_M)}>\frac{M}{2}.
\]
By $M\rightarrow \infty$, this yields the desired result. 

It remains to show $C(\PNE(G_M))>M$.
For fixed $M\geq 1$, let $s=(z,p)$ be a PNE of $G_M$ with induced Wardrop flow $x:=x(s)$ and cost $K:=K(s)$. 
Note that $z_i>0$ holds for $i\in \{1,2\}$, since any PNE has at least two positive capacities. Lemma~\ref{lemma:BR2} together with Theorem~\ref{theo:BR1} yields that, for each firm $i\in \{1,2\}$, the cost $K$ either is optimal for  \hyperlink{P_1}{$(\text{P}_i^1)(s_{-i})$}, or optimal for \hyperlink{P_2}{$(\text{P}_i^2)(s_{-i})$}, and has positive objective function value in both cases. 
Since for each firm $i\in \{1,2\}$, the only candidate for a feasible solution of \hyperlink{P_1}{$(\text{P}_i^1)(s_{-i})$} is $2\sqrt{a_i\gamma_i}+b_i=\sqrt{a_i\gamma_i}+b_i+C_i$ and this yields an objective function value of 0, we get that $K$ is an optimal solution of \hyperlink{P_2}{$(\text{P}_i^2)(s_{-i})$}. In particular this yields, by considering firm 2, that $2M=\sqrt{a_2\gamma_2}+b_2+C_2<K$. Furthermore we get $z_1=\frac{x_1}{K-1}$ and $z_2=\frac{M(1-x_1)}{K-M}$ for the capacities of the two firms. 
Altogether, the desired inequality for $C(s)=C(\PNE(G_M))$ follows:
\begin{align*}
C(\PNE(G_M)) &= \frac{1}{z_1}x_1^2+\frac{M}{z_2}(1-x_1)^2+z_1+Mz_2 
 > (K-1)x_1+(K-M)(1-x_1) \\ 
&=Kx_1-x_1+K -Kx_1-M+Mx_1 \geq K-M> M. 
\end{align*}
 \end{proof}

\section*{Acknowledgements}
This research was funded by the Deutsche Forschungsgemeinschaft (DFG, German Research Foundation) - HA 8041/1-1.

A one-page abstract of this paper appeared in the proceedings of the \emph{15th International Conference on Web and Internet Economics} (\cite{HarksSchedel19}).

\bibliographystyle{abbrv}
\bibliography{master-bib} % if more than one, comma separated

%%%%%%%%%%%%%%%%%
\end{document}